\documentclass[a4paper,12pt]{article}
\title{On multidimensional inverse\\
scattering under the Stark effect}
\author{Tadayoshi ADACHI and Yuta TSUJII\\
{\footnotesize Division of Mathematical and Information Sciences,}\\
{\footnotesize Graduate School of
Human and Environmental Studies,
Kyoto University}\\
{\footnotesize Yoshida-Nihonmatsu-cho, Sakyo-ku, Kyoto-shi, Kyoto 606-8501, Japan}}

\date{ }
\usepackage{amssymb}
\usepackage{amsmath}
\usepackage{amsthm}
\usepackage{mathrsfs}
\usepackage{color}
\usepackage{times}
\theoremstyle{plain}
\newtheorem{thm}{Theorem}[section]
\newtheorem{prop}[thm]{Proposition}
\newtheorem{lem}[thm]{Lemma}

\theoremstyle{definition}

\newtheorem{rem}{Remark}[section]
\numberwithin{equation}{section}
\DeclareMathOperator*{\slim}{s-lim}
\allowdisplaybreaks
\begin{document}
\maketitle

\begin{abstract}
We study one of multidimensional inverse scattering problems
for quantum systems in a constant electric field,
by utilization of the Enss-Weder time-dependent method.
The main purpose of this paper is to propose
some methods of sharpening
key estimates in the analysis, which are much simpler
than those in the previous works.
Our methods give an appropriate class
of short-range potentials which can be determined uniquely
by scattering operators, that seems natural in terms of
direct scattering problems.
\end{abstract}

\section{Introduction}

In this paper, we consider one of inverse scattering problems for
quantum systems in a constant electric field $E\in
\boldsymbol{R}^n$,
by applying the Enss-Weder time-dependent method.
Throughout this paper, we assume that
$n\ge2$, and suppose $E=e_1=(1,0,\ldots,0)$. The Hamiltonian $H$
under consideration is given by
\begin{equation}
H=H_0+V;\quad
H_0=p^2/2-E\cdot x=p^2/2-x_1,
\end{equation}
acting on $L^2(\boldsymbol{R}^n)$, where
$x=(x_1,x_2,\ldots,x_n)=(x_1,x_\perp)\in\boldsymbol{R}^n$ and
$p=-i\nabla=(p_1,p_2,\ldots,p_n)=(p_1,p_\perp)$.
We suppose that the potential $V$ is the multiplication
operator by the real-valued time-independent function $V(x)$.
$H_0$ is called the free Stark Hamiltonian, and
$H$ is called a Stark Hamiltonian.
It is well-known that $H_0$ is essentially self-adjoint on $\mathscr{S}(\boldsymbol{R}^n)$ (see e.g. Avron-Herbst~\cite{AH}).
The self-adjoint realization of $H_0$ is also denoted by $H_0$.
Under a certain appropriate condition on $V$, the self-adjointness
of $H$ can be guaranteed. As is well-known, if $V$ satisfies
a short-range condition under the Stark effect that
$|V(x)|\le C\langle x\rangle^{-\gamma_0}$ holds for some $\gamma_0>1/2$,
then the wave operators
\begin{equation}
W^\pm=\slim_{t\to\pm\infty}e^{itH}e^{-itH_0}
\label{1.2}
\end{equation}
exist (see e.g. \cite{AH}). Here $\langle x\rangle=\sqrt{1+x^2}$.
Then the scattering operator $S=S(V)$ is defined by
\begin{equation}
S=(W^+)^*W^-.
\label{1.3}
\end{equation}
Roughly speaking,
we are interested in the \textit{widest} class of short-range potentials
which can be determined uniquely by scattering operators.
In this paper, via the Enss-Weder time-dependent method
(see Enss-Weder~\cite{EW}), we would like to give a certain 
appropriate class
of short-range potentials which may not necessarily be
the widest one but be very close to it, and
is natural also in terms of
direct scattering problems.

In order to state our results precisely, we
make some preparations:
We assume that $V(x)$ is represented as
a sum of a very short-range part $V^\mathrm{vs}(x)$,
a short-range part $V^\mathrm{s}(x)$ and
a long-range part $V^\mathrm{l}(x)$ under the Stark effect:
\begin{equation}
V(x)=V^\mathrm{vs}(x)+V^\mathrm{s}(x)+V^\mathrm{l}(x).
\end{equation}

We say that
$V^\mathrm{vs}\in\mathscr{V}^\mathrm{vs}$ if $V^\mathrm{vs}(x)$
is a real-valued time-independent function and is
decomposed into a sum of a singular part $V_1^\mathrm{vs}(x)$
and a regular part $V_2^\mathrm{vs}(x)$, and
$V_1^\mathrm{vs}$ is compactly
supported and belongs to $L^{q_0}(\boldsymbol{R}^n)$,
where $q_0$ satisfies that $q_0>n/2$ and $q_0\ge2$,
and $V_2^\mathrm{vs}\in C^0(\boldsymbol{R}^n)$ is bounded in
$\boldsymbol{R}^n$ and satisfies
\begin{equation}
\int_0^\infty\|V_2^\mathrm{vs}(x)F(|x|\ge R)\|_{\mathscr{B}(L^2)}\,dR<\infty.
\end{equation}
Here $F(|x|\ge R)$ is the characteristic function of
$\bigl\{x\in\boldsymbol{R}^n\bigm||x|\ge R\bigr\}$.
This condition yields
\begin{equation}
\int_0^\infty\|V^\mathrm{vs}(x)F(|x|\ge R)(1+K_0)^{-1}\|_{\mathscr{B}(L^2)}\,dR<\infty,
\end{equation}
where
\begin{equation}
K_0=p^2/2=-\Delta/2
\end{equation}
is the free Schr\"odinger operator. As is well-known,
this is equivalent to
\begin{equation}
\int_0^\infty\|V^\mathrm{vs}(x)(1+K_0)^{-1}F(|x|\ge R)\|_{\mathscr{B}(L^2)}\,dR<\infty\label{1.8}
\end{equation}
(see e.g. Reed-Simon~\cite{RS}).
Under the condition \eqref{1.8}, the wave operators
\begin{equation}
\varOmega^\pm=\slim_{t\to\pm\infty}e^{itK}e^{-itK_0}
\label{1.9}
\end{equation}
exist, and are asymptotically complete (see e.g. Enss~\cite{E1}).
Here $K=K_0+V^\mathrm{vs}$ is a
Schr\"odinger operator with a short-range potential $V^\mathrm{vs}$.
Then the scattering operator $\Sigma=\Sigma(V^\mathrm{vs})$ is defined by
\begin{equation}
\Sigma=(\varOmega^+)^*\varOmega^-.
\label{1.10}
\end{equation}
By virtue of the results of \cite{EW}, we know that the following holds:
Let $V_1,\,V_2\in\mathscr{V}^\mathrm{vs}$.
If $\Sigma(V_1)=\Sigma(V_2)$, then $V_1=V_2$.

We say that $V^\mathrm{s}\in\mathscr{V}^\mathrm{s}(\tilde{\gamma}_0,\tilde{\gamma}_1)$
with $\tilde{\gamma}_0\ge1/2$ and $\tilde{\gamma}_1\ge1$ if
$V^\mathrm{s}(x)$
is a real-valued time-independent function,
belongs to $C^1(\boldsymbol{R}^n)$ and
satisfies
\begin{equation}
|(\partial^\beta V^\mathrm{s})(x)|\le C_\beta
\langle x\rangle^{-\gamma_{|\beta|}},\quad |\beta|\le1,
\end{equation}
with some $\gamma_0,\,\gamma_1$ such that
$\tilde{\gamma}_0<\gamma_0\le 1$ and
$\tilde{\gamma}_1<\gamma_1\le1+\gamma_0$.
In Weder~\cite{We1} and Adachi-Maehara~\cite{AM},
$\gamma_0$ and $\gamma_1$ were represented as
$\gamma$ and $1+\alpha$, respectively.
However, in this paper, we will
use the above notation for the sake of clarification of
our theory. If $K$ is equal to $K_0+V^\mathrm{vs}+V^\mathrm{s}$,
then $\varOmega^\pm$ in \eqref{1.9} do not exist generally,
because of the long-range condition $\gamma_0\le1$ for $K_0$.
On the other hand, when $H$ is equal to $H_0+V^\mathrm{vs}+V^\mathrm{s}$,
by virtue of the short-range condition $\gamma_0>1/2$ for $H_0$,
it can be shown that $W^\pm$ in \eqref{1.2} exist, and
are asymptotically complete (see e.g. Herbst~\cite{H} and Yajima~\cite{Ya1}).
Hence, we can introduce $S$ in \eqref{1.3}
instead of $\Sigma$ in \eqref{1.10}, as the scattering operator.
Thus we mainly consider the case where $\tilde{\gamma}_0=1/2$.
We note that if $a<b$, then $\mathscr{V}^\mathrm{s}(\tilde{\gamma}_0,b)\subsetneq\mathscr{V}^\mathrm{s}(\tilde{\gamma}_0,a)$.
The condition $\gamma_1>1$ is
necessary for introducing 
the $v$-dependent Graf-type modifier
$M_{G,v}^\mathrm{s}(t)=e^{-i\int_0^tV^\mathrm{s}(vs+e_1s^2/2)\,ds}$,
which was first introduced in \cite{AM}.
Hence we assume $\tilde{\gamma}_1\ge1$ beforehand.

Finally we say that
$V^\mathrm{l}\in\mathscr{V}_D^\mathrm{l}(\tilde{\gamma}_{D,0})$ with
$\tilde{\gamma}_{D,0}\ge1/4$ if $V^\mathrm{l}(x)$ is
a real-valued time-independent function, belongs to $C^2(\boldsymbol{R}^n)$
and satisfies
\begin{equation}
|(\partial^\beta V^\mathrm{l})(x)|\le C_\beta
\langle x\rangle^{-\gamma_D-|\beta|/2},\quad |\beta|\le2,
\end{equation}
with some $\gamma_D$ such that $\tilde{\gamma}_{D,0}<\gamma_D\le1/2$.
If $H$ is equal to $H_0+V^\mathrm{vs}+V^\mathrm{s}+V^\mathrm{l}$,
then $W^\pm$ in \eqref{1.2} do not exist generally, because of
the long-range condition $\gamma_D\le1/2$ for $H_0$. However,
by virtue of the condition $\gamma_D>1/4$,
the Dollard-type modified wave operators
\begin{equation}
W_D^\pm=\slim_{t\to\pm\infty}e^{itH}e^{-itH_0}e^{-i\int_0^tV^\mathrm{l}(ps+e_1s^2/2)\,ds}
\end{equation}
exist,
and are asymptotically complete (see e.g. Jensen-Yajima~\cite{JY},
White~\cite{W} and Adachi-Tamura~\cite{AT2}).
$\mathscr{V}_D^\mathrm{l}(\tilde{\gamma}_{D,0})$'s
are classes of long-range potentials which are appropriate for
the introduction of $W_D^\pm$.

We first consider the short-range case, that is, the case
where $V^\mathrm{l}=0$.
As mentioned above, then $W^\pm$ in \eqref{1.2}
exist, and $S=S(V)$ in \eqref{1.3} can be defined.
The following result is one of those which
we would like to report on in this paper:

\begin{thm}\label{thm1.1}
Let $V_1,\,V_2\in\mathscr{V}^\mathrm{vs}+
\mathscr{V}^\mathrm{s}(1/2,5/4)$.
If $S(V_1)=S(V_2)$, then $V_1=V_2$.
\end{thm}

Theorem \ref{thm1.1} was first proved by Weder~\cite{We1} for
$\mathscr{V}^\mathrm{vs}+\mathscr{V}^\mathrm{s}(3/4,1)$. However,
the short-range parts $V^\mathrm{s}$'s with $1/2<\gamma_0\le3/4$
cannot be treated by the argument of \cite{We1} unfortunately.
Later Nicoleau~\cite{N1} proved this theorem under
the condition that the short-range parts $V^\mathrm{s}$'s
belong to $C^\infty(\boldsymbol{R}^n)$
and satisfy
\begin{equation}
|(\partial^\beta V^\mathrm{s})(x)|\le C_\beta\langle x\rangle^{-\gamma_0-|\beta|}
\end{equation}
with some $\gamma_0>1/2$, and the additional condition $n\ge3$.
\cite{N1} is the first work which treated the case where $1/2<\gamma_0\le3/4$
and suggested that the possible threshold with respect to $\tilde{\gamma}_0$
is equal to $1/2$.
Here we note that $\gamma_1=1+\gamma_0$ is supposed
in \cite{N1}.
After that, Adachi-Maehara~\cite{AM} proved this theorem
for $\mathscr{V}^\mathrm{vs}+\mathscr{V}^\mathrm{s}(1/2,3/2)$
under the condition not $n\ge3$ but $n\ge2$
(see also Adachi-Kamada-Kazuno-Toratani~\cite{AKKT} as for the case where
time-dependent electric fields
are decaying in $|t|$, and Valencia-Weder~\cite{VW}
as for the many body case in a constant electric field).
In \cite{AM}, for the sake of
relaxing the smoothness condition on $V^\mathrm{s}$'s,
the $v$-dependent Graf-type modifier
$M_{G,v}^\mathrm{s}(t)=e^{-i\int_0^tV^\mathrm{s}(vs+e_1s^2/2)\,ds}$
was introduced instead of the Dollard-type modifier
$e^{-i\int_0^tV^\mathrm{s}(p_\perp s+e_1s^2/2)\,ds}$
which was utilized in \cite{N1}.
Hence, how small $\tilde{\gamma}_1$ of $\mathscr{V}^\mathrm{s}(1/2,\tilde{\gamma}_1)$ in Theorem \ref{thm1.1} can be taken has become a problem
to be studied. In Adachi-Fujiwara-Ishida~\cite{AFI},
which treated also the case where the electric fields are time-dependent,
$\tilde{\gamma}_1$ was taken as $(15-\sqrt{17})/8$,
which is greater than $5/4$, by improving the estimates
in a series of lemmas obtained in \cite{AM}.
And, recently Ishida~\cite{I} stated that
$\tilde{\gamma}_1$ could be taken as $1$.
In the direct scattering theory under the Stark effect,
we often suppose that smooth short-range potentials
$V^\mathrm{s}$'s satisfy
\begin{equation}
|(\partial^\beta V^\mathrm{s})(x)|\le C_\beta\langle x\rangle^{-\gamma_0-|\beta|/2}
\end{equation}
with some $\gamma_0>1/2$ (see e.g. \cite{JY}, \cite{W} and \cite{AT2}).
From this viewpoint, one can expect that
the possible threshold with respect to $\tilde{\gamma}_1$ is
equal to $1/2+1/2=1$. But, since
in \cite{I} some misapplications of the H\"older inequality
were made on key points in the argument unfortunately (for the detail,
see Remark \ref{rem2.1} in \S2), it seems difficult to say
that the results of \cite{I} were obtained rigorously.
Nevertheless, \cite{I} does include a nice device
for improving the results obtained in the previous works.
In this paper, we utilize the
device due to \cite{I} as well as our methods which sharpen
the estimates in a series of useful lemmas in obtaining the main results.
By virtue of these, $\tilde{\gamma}_1$ in Theorem \ref{thm1.1}
can be taken as $5/4$ at most.

We next consider the long-range case, that is,
the case where $V^\mathrm{l}\not=0$.
As mentioned above, since $\tilde{\gamma}_{D,0}\ge1/4$,
if $V^\mathrm{l}\in\mathscr{V}_D^\mathrm{l}(\tilde{\gamma}_{D,0})$, then
the Dollard-type modified wave operators
$W_D^\pm$ exist, and the Dollard-type modified
scattering operator $S_D=S_D(V^\mathrm{l};V^\mathrm{vs}+V^\mathrm{s})
=S_D(V^\mathrm{l};V-V^\mathrm{l})$ is defined by
\begin{equation}
S_D=(W_D^+)^*W_D^-.
\end{equation}
Then we also obtain the following result:

\begin{thm}\label{thm1.2}
Suppose that $V^\mathrm{l}\in\mathscr{V}_D^\mathrm{l}(3/8)$ is given.
Let $V_1,\,V_2\in\mathscr{V}^\mathrm{vs}+\mathscr{V}^\mathrm{s}(1/2,$
$5/4)$.
If $S_D(V^\mathrm{l};V_1)=S_D(V^\mathrm{l};V_2)$, then $V_1=V_2$.
Moreover, any one of the Dollard-type modified
scattering operators $S_D$ determines uniquely the total potential
$V$.
\end{thm}

Theorem \ref{thm1.2} was first proved by \cite{AM} for
$\mathscr{V}^\mathrm{vs}+\mathscr{V}^\mathrm{s}(1/2,3/2)$
under the condition that
$V^\mathrm{l}$ belongs to $C^2(\boldsymbol{R}^n)$
and satisfies
\begin{equation}
|(\partial^\beta V^\mathrm{l})(x)|\le C_\beta\langle
x\rangle^{-\gamma_D-\mu|\beta|},\quad
|\beta|\le2
\end{equation}
with $0<\gamma_D\le1/2$ and $1-\gamma_D<\mu\le1$
(see also \cite{AKKT} as for the case where
time-dependent electric fields
are decaying in $|t|$, and \cite{VW}
as for the many body case in a constant electric field).
The condition $1-\gamma_D<\mu$, that is, $\gamma_D+\mu>1$
yields the existence of the Graf-type modified wave operators
\begin{equation}
W_G^\pm=\slim_{t\to\pm\infty}e^{itH}e^{-itH_0}e^{-i\int_0^t
V^\mathrm{l}(e_1s^2/2)\,ds}
\end{equation}
(see Zorbas~\cite{Zo} and Graf~\cite{Gr})
as well as the Dollard-type wave operators $W_D^\pm$
without the additional condition $\gamma_D>1/4$.
For $V^\mathrm{l}\in\mathscr{V}_D^\mathrm{l}(\tilde{\gamma}_{D,0})$,
$W_G^\pm$ do not exist generally because $\gamma_D+1/2\le1$.
Hence, the case where $V^\mathrm{l}\in\mathscr{V}_D^\mathrm{l}(\tilde{\gamma}_{D,0})$ was not treated in \cite{AM}.
Later Theorem \ref{thm1.2} was proved by \cite{AFI} for
$\mathscr{V}^\mathrm{vs}+\mathscr{V}^\mathrm{s}(1/2,\tilde{\gamma}_1)$
with $\tilde{\gamma}_1=(29-\sqrt{41})/16$, which is greater than
$(15-\sqrt{17})/8$, under the assumption that $V^\mathrm{l}\in\mathscr{V}_D^\mathrm{l}(3/8)$, which is the same as the one in our Theorem \ref{thm1.2}.
And, recently \cite{I} stated $\tilde{\gamma}_1$ could be taken as $1$,
but, as mentioned above,
it seems difficult to say
that the results of \cite{I} were obtained rigorously.
We report that also in Theorem \ref{thm1.2},
$\tilde{\gamma}_1$ can be taken as $5/4$ at most.
\medskip

Now we would like to consider the case where $\tilde{\gamma}_1$
is smaller than $5/4$, in particular, $\tilde{\gamma}_1=1$,
which was considered in \cite{I}. To this end,
we need the $C^2$-regularity of $V^\mathrm{s}$'s,
which is stronger than the $C^1$-regularity of $V^\mathrm{s}$'s
imposed in Theorems \ref{thm1.1} and \ref{thm1.2}.
We will introduce the following subclasses of $\mathscr{V}^\mathrm{s}(1/2,\tilde{\gamma}_1)$:
\smallskip

We say that
$V^\mathrm{s}\in\tilde{\mathscr{V}}^\mathrm{s}(1/2,\tilde{\gamma}_1,\tilde{\gamma}_2)$ with $\tilde{\gamma}_1\ge1$ and
$\tilde{\gamma}_2\ge1$ if $V^\mathrm{s}(x)$ is
a real-valued time-independent function, belongs to $C^2(\boldsymbol{R}^n)$
and satisfies
\begin{equation}
|(\partial^\beta V^\mathrm{s})(x)|\le C_\beta
\langle x\rangle^{-\gamma_{|\beta|}},\quad |\beta|\le2,
\end{equation}
with some $\gamma_0,\,\gamma_1,\,\gamma_2$ such that
$\tilde{\gamma}_0=1/2<\gamma_0\le1$,
$\tilde{\gamma}_1<\gamma_1\le\gamma_0+1$
and $\tilde{\gamma}_2<\gamma_2\le\gamma_1+1$.
Since we need the condition $\gamma_2>1$ in our analysis (see \S4),
we assume $\tilde{\gamma}_2\ge1$ beforehand.

Now we would like to state that
Theorems \ref{thm1.1} and \ref{thm1.2}
with replacing $\mathscr{V}^\mathrm{s}(1/2,$
$5/4)$
by $\tilde{\mathscr{V}}^\mathrm{s}(1/2,1,5/4)$ also hold:

\begin{thm}\label{thm1.3}
Let $V_1,\,V_2\in\mathscr{V}^\mathrm{vs}+\tilde{\mathscr{V}}^\mathrm{s}(1/2,1,5/4)$.
If $S(V_1)=S(V_2)$, then $V_1=V_2$.
\end{thm}

\begin{thm}\label{thm1.4}
Suppose that $V^\mathrm{l}\in\mathscr{V}_D^\mathrm{l}(3/8)$ is given.
Let $V_1,\,V_2\in\mathscr{V}^\mathrm{vs}+\tilde{\mathscr{V}}^\mathrm{s}(1/2,1,\\
5/4)$.
If $S_D(V^\mathrm{l};V_1)=S_D(V^\mathrm{l};V_2)$, then $V_1=V_2$.
Moreover, any one of the Dollard-type modified
scattering operators $S_D$ determines uniquely the total potential
$V$.
\end{thm}

In their proofs, the Dollard-type
modifier $M_D^\mathrm{s}(t)=e^{-i\int_0^tV^\mathrm{s}(ps+e_1s^2/2)\,ds}$,
which is slightly different from
$e^{-i\int_0^tV^\mathrm{s}(p_\perp s+e_1s^2/2)\,ds}$ introduced in \cite{N1},
will be substituted for the $v$-dependent Graf-type
modifier $M_{G,v}^\mathrm{s}(t)=e^{-i\int_0^tV^\mathrm{s}(vs+e_1s^2/2)\,ds}$
used in the proofs of Theorems \ref{thm1.1} and \ref{thm1.2}.
$M_D^\mathrm{s}(t)$ has an advantage
over $M_{G,v}^\mathrm{s}(t)$ for $V^\mathrm{s}$ of $C^2$
by virtue of the Baker-Campbell-Hausdorff formula,
although it may not necessarily do so for $V^\mathrm{s}$ of only $C^1$.
We will report on it in this paper.
\bigskip

The plan of this paper is as follows: In \S2, we consider the case
where $V^\mathrm{s}\in\mathscr{V}^\mathrm{s}(1/2,5/4)$ and
$V^\mathrm{l}=0$.
In \S3, we consider the case
where $V^\mathrm{s}\in\mathscr{V}^\mathrm{s}(1/2,5/4)$
and $V^\mathrm{l}\not=0$.
In \S4, we consider the case where $V^\mathrm{s}\in\tilde{\mathscr{V}}^\mathrm{s}(1/2,1,5/4)$.

In the following sections, $\|\cdot\|$ and $(\cdot,\cdot)$ stand for
the $L^2$-norm and the $L^2$-inner product, respectively,
for the sake of brevity.

\section{The case where $V^\mathrm{s}\in\mathscr{V}^\mathrm{s}(1/2,5/4)$ and $V^\mathrm{l}=0$}

Throughout this section, we suppose
$V^\mathrm{s}\in\mathscr{V}^\mathrm{s}(1/2,1)$
and $V^\mathrm{l}=0$.
The main purpose of this section is showing the following reconstruction
formula, which yields Theorem \ref{thm1.1}.

\begin{thm}\label{thm2.1} Let $\hat{v}\in\boldsymbol{R}^n$ be given
such that $|\hat{v}\cdot e_1|<1$. Put $v=|v|\hat{v}$. Let $\eta>0$ be given,
and $\varPhi_0,\,\varPsi_0\in L^2(\boldsymbol{R}^n)$ be
such that $\hat{\varPhi}_0,\,\hat{\varPsi}_0\in
C_0^\infty(\boldsymbol{R}^n)$ with
$\mathrm{supp}\,\hat{\varPhi}_0,\,
\mathrm{supp}\,\hat{\varPsi}_0\subset
\bigl\{\xi\in\boldsymbol{R}^n\bigm||\xi|<\eta\bigr\}$.
Put $\varPhi_v=e^{iv\cdot x}\varPhi_0$ and $\varPsi_v=e^{iv\cdot x}\varPsi_0$.
Let $V^\mathrm{vs}\in\mathscr{V}^\mathrm{vs}$ and
$V^\mathrm{s}\in\mathscr{V}^\mathrm{s}(1/2,5/4)$.
Then the following holds:
\begin{equation}
\begin{split}
\lim_{|v|\to\infty}{}&|v|(i[S,p_j]\varPhi_v,\varPsi_v)\\
{}&=\int_{-\infty}^\infty[
(V^\mathrm{vs}(x+\hat{v}\tau)p_j\varPhi_0,\varPsi_0)-
(V^\mathrm{vs}(x+\hat{v}\tau)\varPhi_0,p_j\varPsi_0)\\
{}&\qquad\qquad+i((\partial_j V^\mathrm{s})(x+\hat{v}\tau)\varPhi_0,\varPsi_0)]\,d\tau
\end{split}\label{2.1}
\end{equation}
for $1\le j\le n$.
\end{thm}

We will make preparations for the proof of Theorem \ref{thm2.1}.
We first need the following proposition due
to Enss~\cite{E2} (see Proposition 2.10 of \cite{E2}):

\begin{prop}\label{prop2.2}
For any $f\in C_0^\infty(\boldsymbol{R}^n)$ with $\mathrm{supp}\,f\subset
\bigl\{x\in\boldsymbol{R}^n\bigm||x|<\eta\bigr\}$ for some $\eta>0$,
and any $l\in\boldsymbol{N}$, there exists a constant $C_l$
dependent on $f$ only such that
\begin{equation}
\|F(x\in\mathscr{M}')e^{-itK_0}f(p-v)F(x\in\mathscr{M})\|_{\mathscr{B}(L^2)}
\le C_l(1+r+|t|)^{-l}
\end{equation}
for $v\in\boldsymbol{R}^n$, $t\in\boldsymbol{R}$ and measurable sets
$\mathscr{M}$, $\mathscr{M}'$ with the property that
$r=\mathrm{dist}(\mathscr{M}',\mathscr{M}+vt)-\eta|t|\ge0$.
Here $F(x\in\mathscr{M})$ stands for the characteristic function
of $\mathscr{M}$.
\end{prop}

The following lemma was already obtained in \cite{We1}
(see also \cite{AM}):

\begin{lem}\label{lem2.3} Let $v$ and $\varPhi_v$ be as in Theorem \ref{thm2.1}. Then
\begin{equation}
\int_{-\infty}^\infty
\|V^\mathrm{vs}(x)e^{-itH_0}\varPhi_v\|\,dt
=O(|v|^{-1})
\end{equation}
holds as $|v|\to\infty$ for $V^\mathrm{vs}\in\mathscr{V}^\mathrm{vs}$.
\end{lem}

In the proof of this lemma, the estimate of $|vt+e_1t^2/2|$ plays an important
role. Here we recall the argument about the estimate of $|vt+e_1t^2/2|$
in \cite{AM}:
Put $\delta=|\hat{v}\cdot e_1|<1$.
$|vt+e_1t^2/2|^2=|v|^2t^2+t^4/4+v\cdot e_1t^3$
can be estimated as
\begin{equation}
\begin{split}
|vt+e_1t^2/2|^2\ge{}&|v|^2|t|^2+|t|^4/4-\delta|v||t|^3\\
={}&|t|^2(|t|-2\delta|v|)^2/4+(1-\delta^2)|v|^2|t|^2\\
\ge{}&(1-\delta^2)|v|^2|t|^2.
\end{split}
\label{2.4}
\end{equation}
\eqref{2.4} is used in the proof of Lemma \ref{lem2.3}.
Here we note that
\begin{equation}
\begin{split}
|vt+e_1t^2/2|^2\ge{}&|t|^2(\delta|t|-2|v|)^2/4+(1-\delta^2)|t|^4/4\\
\ge{}&(1-\delta^2)|t|^4/4
\end{split}
\label{2.5}
\end{equation}
can be also obtained. Based on the above estimates, we conclude that
\begin{equation}
|vt+e_1t^2/2|\ge\max\{\sqrt{1-\delta^2}|v||t|,(\sqrt{1-\delta^2}/2)|t|^2\}
\label{2.6}
\end{equation}
holds.
\medskip

The following lemma is an improvement of Lemma 2.2 of
\cite{AM} and Lemma 3.4 with $\mu=0$ of \cite{AFI}.
This is one of the keys in this section:

\begin{lem}\label{lem2.4} Let $v$ and $\varPhi_v$
be as in Theorem \ref{thm2.1}, and $\epsilon>0$.
Then
\begin{equation}
\int_{-\infty}^\infty
\|\{V^\mathrm{s}(x)-V^\mathrm{s}(vt+e_1t^2/2)\}e^{-itH_0}\varPhi_v\|\,dt
=O(|v|^{\max\{-1,-2(\gamma_1-1)+\epsilon\}})
\label{2.7}
\end{equation}
holds as $|v|\to\infty$ for $V^\mathrm{s}\in\mathscr{V}^\mathrm{s}(1/2,1)$.
\end{lem}

\begin{proof}
For the sake of brevity, we put $I=
\|\{V^\mathrm{s}(x)-V^\mathrm{s}(vt+e_1t^2/2)\}e^{-itH_0}\varPhi_v\|$.
By virtue of the Avron-Herbst formula
\begin{equation}
e^{-itH_0}=e^{-it^3/6}e^{itx_1}e^{-ip_1t^2/2}e^{-itK_0}
\label{2.8}
\end{equation}
(see e.g. \cite{AH}) and
\begin{equation}
e^{-iv\cdot x}e^{-itK_0}e^{iv\cdot x}=
e^{-iv^2t/2}e^{-ip\cdot vt}e^{-itK_0},
\label{2.9}
\end{equation}
we have
\begin{align*}
&V^\mathrm{s}(x)e^{-itH_0}\varPhi_v\\
={}&(e^{-it^3/6}e^{itx_1}e^{-ip_1t^2/2})V^\mathrm{s}(x+e_1t^2/2)e^{-itK_0}\varPhi_v\\
={}&(e^{-it^3/6}e^{itx_1}e^{-ip_1t^2/2})e^{iv\cdot x}(e^{-iv^2t/2}e^{-ip\cdot vt})
V^\mathrm{s}(x+vt+e_1t^2/2)e^{-itK_0}\varPhi_0.
\end{align*}
Such a relation has been used also in the proof of Lemma \ref{lem2.3},
which is omitted in this paper.
Then $I$ can be written as
$$I=\|\{V^\mathrm{s}(x+vt+e_1t^2/2)-V^\mathrm{s}(vt+e_1t^2/2)\}e^{-itK_0}\varPhi_0\|.$$
Taking $f\in C_0^\infty(\boldsymbol{R}^n)$ such that $0\le f\le1$,
$f\hat{\varPhi}_0=\hat{\varPhi}_0$ and $\mathrm{supp}\,f\subset
\bigl\{\xi\in\boldsymbol{R}^n\bigm||\xi|<\eta\bigr\}$,
we see that $\varPhi_0=f(p)\varPhi_0$.
Now let us take $g\in C_0^\infty(\boldsymbol{R}^n)$ such that
$0\le g\le1$, $g(y)=1\ (|y|\le3)$ and $g(y)=0\ (|y|\ge4)$,
and introduce
$$\tilde{V}_{|v|,t}^\mathrm{s}(x)=V^\mathrm{s}(x+vt+e_1t^2/2)
g((\lambda_1|v|\langle t\rangle)^{-1}x),$$
where $\lambda_1>0$ is a small constant which will be determined below.
In the definition of $\tilde{V}_{|v|,t}^\mathrm{s}(x)$,
taking $(\lambda_1|v|\langle t\rangle)^{-1}x$
as the argument of $g$,
which was taken as $(\lambda_1|v||t|)^{-1}x$ in \cite{AM},
is one of our devices,
and by virtue of this, we can eliminate the singularity of $\tilde{V}_{|v|,t}^\mathrm{s}(x)$ and its derivatives at $t=0$. Since
$\tilde{V}_{|v|,t}^\mathrm{s}(0)=V^\mathrm{s}(vt+e_1t^2/2)$
and $\varPhi_0=f(p)\varPhi_0$, we have
\begin{equation}
I\le\|\bar{V}_{|v|,t}^\mathrm{s}(x)e^{-itK_0}
f(p)\varPhi_0\|+
\|\{\tilde{V}_{|v|,t}^\mathrm{s}(x)-\tilde{V}_{|v|,t}^\mathrm{s}(0)\}e^{-itK_0}
\varPhi_0\|,\label{2.10}
\end{equation}
where $\bar{V}_{|v|,t}^\mathrm{s}(x)=V^\mathrm{s}(x+vt+e_1t^2/2)-\tilde{V}_{|v|,t}^\mathrm{s}(x)$.
As for the first term of the inequality \eqref{2.10},
we estimate it in the same way as in \cite{We1}:
$$\|\bar{V}_{|v|,t}^\mathrm{s}(x)e^{-itK_0}
f(p)\varPhi_0\|\le
I_1+I_2+I_3,$$
where
\begin{align*}
I_1={}&
\|\bar{V}_{|v|,t}^\mathrm{s}(x)F(|x|\ge3\lambda_1|v||t|)e^{-itK_0}f(p)F(|x|\le\lambda_1|v||t|)
\varPhi_0\|,\\
I_2={}&
\|\bar{V}_{|v|,t}^\mathrm{s}(x)F(|x|\ge3\lambda_1|v||t|)e^{-itK_0}f(p)F(|x|>\lambda_1|v||t|)
\varPhi_0\|,\\
I_3={}&
\|\bar{V}_{|v|,t}^\mathrm{s}(x)F(|x|<3\lambda_1|v||t|)e^{-itK_0}f(p)\varPhi_0\|.
\end{align*}
As for $I_1$,
by virtue of Proposition \ref{prop2.2}
and $\|\bar{V}_{|v|,t}^\mathrm{s}(x)\|_{\mathscr{B}(L^2)}
\le\|V^\mathrm{s}\|_{L^\infty}$,
\begin{align*}
I_1\le{}&\|V^\mathrm{s}\|_{L^\infty}\|F(|x|\ge3\lambda_1|v||t|)e^{-itK_0}f(p)F(|x|\le\lambda_1|v||t|)\|_{\mathscr{B}(L^2)}
\|\varPhi_0\|\\
\le{}&C(1+\lambda_1|v||t|)^{-2}
\end{align*}
holds for $|v|>\eta/\lambda_1$.
Here we used $0\le1-g\le1$. As for $I_2$,
we have
\begin{align*}
I_2\le{}&\|V^\mathrm{s}\|_{L^\infty}
\|F(|x|>\lambda_1|v||t|)\langle x\rangle^{-2}\|_{\mathscr{B}(L^2)}
\|\langle x\rangle^2\varPhi_0\|\\
\le{}&C(1+\lambda_1|v||t|)^{-2}.
\end{align*}
As for $I_3$, by virtue of the definition of
$\tilde{V}_{|v|,t}^\mathrm{s}(x)$, we have $I_3=0$,
because $\{1-g((\lambda_1|v|\langle t\rangle)^{-1}x)\}
F(|x|<3\lambda_1|v||t|)\equiv0$.
Based on the above observations, we obtain
\begin{equation}
\int_{-\infty}^\infty\|\bar{V}_{|v|,t}^\mathrm{s}(x)e^{-itK_0}
f(p)\varPhi_0\|\,dt
\le C\int_{-\infty}^\infty(1+\lambda_1|v||t|)^{-2}\,dt
=O(|v|^{-1})
\label{2.11}
\end{equation}
by an appropriate change of variables.
Now we will estimate
the second term of the inequality \eqref{2.10}
by using the device of \cite{I}. By
$$
\tilde{V}_{|v|,t}^\mathrm{s}(x)-\tilde{V}_{|v|,t}^\mathrm{s}(0)
=\int_0^1\frac{d}{d\theta}\{\tilde{V}_{|v|,t}^\mathrm{s}(\theta x)\}\,d\theta
=\left(\int_0^1(\nabla\tilde{V}_{|v|,t}^\mathrm{s})(\theta x)\,d\theta\right)\cdot x,
$$
we have
\begin{align*}
\|\{\tilde{V}_{|v|,t}^\mathrm{s}(x)-\tilde{V}_{|v|,t}^\mathrm{s}(0)\}e^{-itK_0}
\varPhi_0\|
\le{}&\left\|\int_0^1(\nabla\tilde{V}_{|v|,t}^\mathrm{s})(\theta x)\,d\theta\right\|_{\mathscr{B}(L^2)}\|xe^{-itK_0}\varPhi_0\|\\
\le{}&\|\nabla\tilde{V}_{|v|,t}^\mathrm{s}\|_{L^\infty}\|(x+pt)\varPhi_0\|\\
\le{}&\|\nabla\tilde{V}_{|v|,t}^\mathrm{s}\|_{L^\infty}(\|x\varPhi_0\|+|t|\|p\varPhi_0\|).
\end{align*}
Here we used $e^{itK_0}xe^{-itK_0}=x+pt$.
Dealing with $\|xe^{-itK_0}\varPhi_0\|$ directly without using cut-offs is the device of \cite{I}.
Now we will watch $\|\nabla\tilde{V}_{|v|,t}^\mathrm{s}\|_{L^\infty}$, where
\begin{align*}
(\nabla\tilde{V}_{|v|,t}^\mathrm{s})(x)
={}&(\nabla V^\mathrm{s})(x+vt+e_1t^2/2)
g((\lambda_1|v|\langle t\rangle)^{-1}x)\\
&\qquad+V^\mathrm{s}(x+vt+e_1t^2/2)
(\nabla g)((\lambda_1|v|\langle t\rangle)^{-1}x)(\lambda_1|v|\langle t\rangle)^{-1}.
\end{align*}
For a while, suppose $|t|\ge1$. Then $\langle t\rangle
=\sqrt{1+t^2}\le\sqrt{t^2+t^2}=\sqrt{2}|t|$ holds.
Note that $x\in\mathrm{supp}\,\{g(\cdot/(\lambda_1|v|\langle t\rangle))\}$
satisfies $|x|\le4\lambda_1|v|\langle t\rangle$. For such $x$'s,
\begin{align*}
|x+vt+e_1t^2/2|^2={}&|x+vt|^2+t^4/4+t^2(x+vt)\cdot e_1\\
\ge{}&(|vt|-|x|)^2+|t|^4/4-|t|^2(|x_1|+\delta|v||t|)\\
\ge{}&(|v||t|-4\lambda_1|v|\langle t\rangle)^2+|t|^4/4-|t|^2(4\lambda_1|v|\langle t\rangle+\delta|v||t|)\\
\ge{}&\{(1-4\sqrt{2}\lambda_1)|v||t|\}^2+|t|^4/4-|t|^2\{(4\sqrt{2}\lambda_1+\delta)|v||t|\}
\end{align*}
holds. If $\lambda_1$ is so small that $4\sqrt{2}\lambda_1\le(1-\delta)/4$,
then
\begin{align*}
|x+vt+e_1t^2/2|^2
\ge{}&((3+\delta)/4)^2|v|^2|t|^2+|t|^4/4-((1+3\delta)/4)|v||t|^3\\
={}&|t|^2\{|t|-((1+3\delta)/2)|v|\}^2/4+((1-\delta^2)/2)|v|^2|t|^2\\
\ge{}&((1-\delta^2)/2)|v|^2|t|^2
\end{align*}
holds (cf. \eqref{2.4}). We also have
\begin{align*}
|x+vt+e_1t^2/2|^2
\ge{}&|t|^2(((1+3\delta)/\{2(3+\delta)\})|t|-((3+\delta)/4)|v|)^2\\
&\qquad+(2(1-\delta^2)/(3+\delta)^2)|t|^4\\
\ge{}&(2(1-\delta^2)/(3+\delta)^2)|t|^4
\end{align*}
(cf. \eqref{2.5}). Hence, by taking $\lambda_1$ as
$(1-\delta)/(16\sqrt{2})$,
\begin{equation}
|x+vt+e_1t^2/2|\ge
\max\{c_1|v||t|,c_2|t|^2\}\label{2.12}
\end{equation}
holds for $|t|\ge1$, where
$c_1=\sqrt{2(1-\delta^2)}/2$ and
$c_2=\sqrt{2(1-\delta^2)}/(3+\delta)$. \eqref{2.12} yields
\begin{equation}
|x+vt+e_1t^2/2|\ge
(c_1|v||t|)^\nu(c_2|t|^2)^{1-\nu}=c_1^\nu c_2^{1-\nu}|v|^\nu|t|^{2-\nu}
\label{2.13}
\end{equation}
for $0\le\nu\le1$.
Throughout this paper, we will frequently use this interpolation
estimate on $|x+vt+e_1t^2/2|$ parameterized by $\nu$'s such that
$0\le\nu\le1$ in our optimization arguments.
Taking account of the boundedness of 
$\nabla\tilde{V}_{|v|,t}^\mathrm{s}$ for $|t|<1$,
we see that
\begin{equation}
\begin{split}
\|\nabla\tilde{V}_{|v|,t}^\mathrm{s}\|_{L^\infty}
\le{}&
C_1(1+|v|^{\nu_1}|t|^{2-\nu_1})^{-\gamma_1}\\
&+C_2(1+|v|^{\nu_2}|t|^{2-\nu_2})^{-\gamma_0}|v|^{-1}(1+|t|)^{-1}\\
\le{}&
C_1'(1+|v|^{\nu_1/(2-\nu_1)}|t|)^{-\gamma_1(2-\nu_1)}\\
&+C_2'(1+|v|^{\nu_2/(2-\nu_2)}|t|)^{-\gamma_0(2-\nu_2)}\\
&\qquad\quad\times|v|^{\nu_2/(2-\nu_2)-1}(|v|^{\nu_2/(2-\nu_2)}+|v|^{\nu_2/(2-\nu_2)}|t|)^{-1}\\
\le{}&
C_1'(1+|v|^{\nu_1/(2-\nu_1)}|t|)^{-\gamma_1(2-\nu_1)}\\
&+C_2'|v|^{\nu_2/(2-\nu_2)-1}
(1+|v|^{\nu_2/(2-\nu_2)}|t|)^{-\gamma_0(2-\nu_2)-1}
\end{split}\label{2.14}
\end{equation}
holds for $0\le\nu_1,\,\nu_2\le1$ and $|v|\ge1$.
Then
\begin{align*}
&\int_{-\infty}^\infty\|\nabla\tilde{V}_{|v|,t}^\mathrm{s}\|_{L^\infty}
\|x\varPhi_0\|\,dt\\
\le{}&C_1'\int_{-\infty}^\infty(1+|v|^{\nu_1/(2-\nu_1)}|t|)^{-\gamma_1(2-\nu_1)}\,dt\\
&+C_2'|v|^{\nu_2/(2-\nu_2)-1}\int_{-\infty}^\infty(1+|v|^{\nu_2/(2-\nu_2)}|t|)^{-\gamma_0(2-\nu_2)-1}\,dt\\
={}&O(|v|^{-\nu_1/(2-\nu_1)})+
O(|v|^{\nu_2/(2-\nu_2)-1-\nu_2/(2-\nu_2)})\\
={}&O(|v|^{-\nu_1/(2-\nu_1)})+
O(|v|^{-1})
\end{align*}
can be obtained by an appropriate change of variables
under the conditions
$-\gamma_1(2-\nu_1)<-1$ and $-\gamma_0(2-\nu_2)-1<-1$,
i.e. $\nu_1<2-1/\gamma_1$
and $\nu_2<2$.
Since $2-1/\gamma_1>1$ and
$$\min_{0\le\nu_1\le1}(-\nu_1/(2-\nu_1))=-1,$$
we have
\begin{equation}
\int_{-\infty}^\infty\|\nabla\tilde{V}_{|v|,t}^\mathrm{s}\|_{L^\infty}
\|x\varPhi_0\|\,dt
=O(|v|^{-1})
\label{2.15}
\end{equation}
by the optimization argument. In the same way,
\begin{align*}
&\int_{-\infty}^\infty\|\nabla\tilde{V}_{|v|,t}^\mathrm{s}\|_{L^\infty}
|t|\|p\varPhi_0\|\,dt\\
\le{}&C_3'\int_{-\infty}^\infty(1+|v|^{\nu_3/(2-\nu_3)}|t|)^{-\gamma_1(2-\nu_3)}|t|\,dt\\
&+C_4'|v|^{\nu_4/(2-\nu_4)-1}\int_{-\infty}^\infty(1+|v|^{\nu_4/(2-\nu_4)}|t|)^{-\gamma_0(2-\nu_4)-1}|t|\,dt\\
={}&O(|v|^{-2\nu_3/(2-\nu_3)})+
O(|v|^{\nu_4/(2-\nu_4)-1-2\nu_4/(2-\nu_4)})\\
={}&O(|v|^{-2\nu_3/(2-\nu_3)})+
O(|v|^{-1-\nu_4/(2-\nu_4)})
\end{align*}
can be obtained by an appropriate change of variables
under the conditions
$-\gamma_1(2-\nu_3)+1<-1$ and $-\gamma_0(2-\nu_4)-1+1<-1$,
i.e. $\nu_3<2-2/\gamma_1$
and $\nu_4<2-1/\gamma_0$.
Since $0<2-2/\gamma_1\le2-2/(1+\gamma_0)<2/3$,
$0<2-1/\gamma_0\le1$,
\begin{align*}
&\inf_{0\le\nu_3<2-2/\gamma_1}(-2\nu_3/(2-\nu_3))=-2(\gamma_1-1),\\
&\inf_{0\le\nu_4<2-1/\gamma_0}(-1-\nu_4/(2-\nu_4))=-2\gamma_0,
\end{align*}
and $-2\gamma_0\le-2(\gamma_1-1)$, we have
\begin{equation}
\int_{-\infty}^\infty\|\nabla\tilde{V}_{|v|,t}^\mathrm{s}\|_{L^\infty}
|t|\|p\varPhi_0\|\,dt=O(|v|^{-2(\gamma_1-1)+\epsilon})
\label{2.16}
\end{equation}
with $\epsilon>0$, by the optimization argument.
\eqref{2.15} and \eqref{2.16} yield
\begin{equation}
\int_{-\infty}^\infty\|\{\tilde{V}_{|v|,t}^\mathrm{s}(x)-\tilde{V}_{|v|,t}^\mathrm{s}(0)\}e^{-itK_0}
\varPhi_0\|\,dt=O(|v|^{-1})+O(|v|^{-2(\gamma_1-1)+\epsilon}).
\label{2.17}
\end{equation}
Therefore, \eqref{2.7} can be shown by \eqref{2.11} and \eqref{2.17}.
\end{proof}

\begin{rem}\label{rem2.1}
Taking $(\lambda_1|v|\langle t\rangle)^{-1}x$ in place of
$(\lambda_1|v||t|)^{-1}x$ as the argument of $g$
in the definition of $\tilde{V}_{|v|,t}^\mathrm{s}(x)$
yields the regularity of $(\nabla\tilde{V}_{|v|,t}^\mathrm{s})(x)$
also at $t=0$. Hence, we do not have to divide
the integral region $\boldsymbol{R}$ of $t$ into any $v$-dependent
neighborhood of $0$ and its complement (see below) as in \cite{We1},
\cite{AM}, \cite{AFI}, \cite{I} and so on.
This fact makes the optimization argument in the
above proof rather simple.

Here we would like to review the argument
in the proof of Proposition 2.4 of \cite{I}
corresponding to our Lemma \ref{lem2.4}: Let $0<\sigma_1<1$.
Taking account of
$$\int_{-\infty}^\infty I\,dt=
\int_{|t|<|v|^{-\sigma_1}}I\,dt+\int_{|t|\ge|v|^{-\sigma_1}}I\,dt$$
with $I$ in the proof of our Lemma \ref{lem2.4},
we have only to estimate the two terms of the
right-hand side separately, as mentioned above. $I$ can be estimated as
$I\le \hat{I}_1+\hat{I}_2+\hat{I}_3$, where
$\hat{I}_j$'s are defined by replacing $\bar{V}_{|v|,t}^\mathrm{s}(x)$
in the definition of $I_j$'s by
$\hat{V}_{v,t}^\mathrm{s}(x)=V^\mathrm{s}(x+vt+e_1t^2/2)-V^\mathrm{s}(vt+e_1t^2/2)$:
\begin{align*}
\hat{I}_1={}&
\|\hat{V}_{v,t}^\mathrm{s}(x)F(|x|\ge3\lambda_1|v||t|)e^{-itK_0}f(p)F(|x|\le\lambda_1|v||t|)
\varPhi_0\|,\\
\hat{I}_2={}&
\|\hat{V}_{v,t}^\mathrm{s}(x)F(|x|\ge3\lambda_1|v||t|)e^{-itK_0}f(p)F(|x|>\lambda_1|v||t|)
\varPhi_0\|,\\
\hat{I}_3={}&
\|\hat{V}_{v,t}^\mathrm{s}(x)F(|x|<3\lambda_1|v||t|)e^{-itK_0}f(p)\varPhi_0\|.
\end{align*}
$\hat{I}_1+\hat{I}_2\le C(1+\lambda_1|v||t|)^{-2}$ yields
$$\int_{|t|<|v|^{-\sigma_1}}(\hat{I}_1+\hat{I}_2)\,dt=O(|v|^{-1}).$$
On the other hand, $\hat{I}_3\le C(1+\lambda_1|v||t|)^{-\gamma_1}(1+|t|)$
yields
$$\int_{|t|<|v|^{-\sigma_1}}\hat{I}_3\,dt=O(|v|^{-1})
+o(|v|^{-1})=O(|v|^{-1}).$$
Let $0<\sigma_2<1$.
$I$ can be also estimated as
$I\le \hat{I}_4+\hat{I}_5+\hat{I}_6$, where
\begin{align*}
\hat{I}_4={}&
\|\hat{V}_{v,t}^\mathrm{s}(x)F(|x|\ge3|v|^{\sigma_2}|t|)e^{-itK_0}f(p)F(|x|\le|v|^{\sigma_2}|t|)
\varPhi_0\|,\\
\hat{I}_5={}&
\|\hat{V}_{v,t}^\mathrm{s}(x)F(|x|\ge|v|^{\sigma_2}|t|)e^{-itK_0}f(p)F(|x|>|v|^{\sigma_2}|t|)
\varPhi_0\|,\\
\hat{I}_6={}&
\|\hat{V}_{v,t}^\mathrm{s}(x)F(|x|<3|v|^{\sigma_2}|t|)e^{-itK_0}f(p)\varPhi_0\|.
\end{align*}
$\hat{I}_4+\hat{I}_5\le C(1+|v|^{\sigma_2}|t|)^{-2}$ yields
$$\int_{|t|\ge|v|^{-\sigma_1}}(\hat{I}_4+\hat{I}_5)\,dt=O(|v|^{\sigma_1-2\sigma_2}).$$
Here we suppose that $|v|$ is so large that $4\sqrt{2}|v|^{\sigma_2-1}\le (1-\delta)/4$
holds. Then $\hat{I}_6\le C(1+c_1|v||t|)^{-\gamma_1}(1+|t|)$
and $\hat{I}_6\le C(1+c_2|t|^2)^{-\gamma_1}(1+|t|)$
hold with $c_1$ and $c_2$ in the proof of our Lemma \ref{lem2.4}.
These estimates can be obtained rigorously.
In \cite{I}, Ishida stated that the H\"older inequality yields
the estimate (2.44) in the proof of Proposition 2.4 of \cite{I}
$$\int_{|t|\ge|v|^{-\sigma_1}}\hat{I}_6\,dt\le
\left(\int_{|t|\ge|v|^{-\sigma_1}}\hat{I}_6^{q_1}\,dt\right)^{1/q_1}
\left(\int_{|t|\ge|v|^{-\sigma_1}}\hat{I}_6^{q_2}\,dt\right)^{1/q_2}$$
with $q_1>1$ such that $1/q_1+1/q_2=1$ (see also (3.46) in the proof of Proposition 3.8 of \cite{I}).
However, an appropriate application of the H\"older inequality yields
a trivial estimate
$$\int_{|t|\ge|v|^{-\sigma_1}}\hat{I}_6\,dt\le
\left(\int_{|t|\ge|v|^{-\sigma_1}}\hat{I}_6\,dt\right)^{1/q_1}
\left(\int_{|t|\ge|v|^{-\sigma_1}}\hat{I}_6\,dt\right)^{1/q_2}$$
only, because $\hat{I}_6=
\hat{I}_6^{1/q_1}\times\hat{I}_6^{1/q_2}$.
Therefore we think that the results in Propositions 2.4 and 3.8 of
\cite{I} have not been
obtained rigorously yet. Following our argument, we would like
to use the estimate
$$\hat{I}_6\le C_1(1+c_1^{\nu_1}c_2^{1-\nu_1}|v|^{\nu_1}|t|^{2-\nu_1})^{-\gamma_1}+C_2(1+c_1^{\nu_2}c_2^{1-\nu_2}|v|^{\nu_2}|t|^{2-\nu_2})^{-\gamma_1}|t|$$
for $0\le\nu_1,\,\nu_2\le1$. As for $\nu_1$,
$-\gamma_1(2-\nu_1)\le-\gamma_1<-1$
holds; while, as for $\nu_2$,
we need that $\nu_2$ satisfies $-\gamma_1(2-\nu_2)+1<-1$, that is,
$\nu_2<2-2/\gamma_1$. For the sake of simplicity,
we assume $2-2/\gamma_1\le1$, that is, $\gamma_1\le2$.
Then we have the estimate
\begin{align*}
\int_{|t|\ge|v|^{-\sigma_1}}\hat{I}_6\,dt={}&
O(|v|^{-\nu_1\gamma_1-\sigma_1\{1-(2-\nu_1)\gamma_1\}})+
O(|v|^{-\nu_2\gamma_1-\sigma_1\{2-(2-\nu_2)\gamma_1\}})\\
={}&
O(|v|^{-\sigma_1(1-2\gamma_1)-(1+\sigma_1)\nu_1\gamma_1})+
O(|v|^{-\sigma_1(2-2\gamma_1)-(1+\sigma_1)\nu_2\gamma_1}).
\end{align*}
Taking $\nu_1=1$ and $\nu_2=2-2/\gamma_1$, we obtain
a finer estimate
\begin{align*}
\int_{|t|\ge|v|^{-\sigma_1}}\hat{I}_6\,dt={}&
O(|v|^{-\sigma_1(1-2\gamma_1)-(1+\sigma_1)\gamma_1})+
O(|v|^{-\sigma_1(2-2\gamma_1)-(1+\sigma_1)(2-2/\gamma_1)\gamma_1+\epsilon})\\
={}&
O(|v|^{-\gamma_1-\sigma_1(1-\gamma_1)})+
O(|v|^{-(2\gamma_1-2)+\epsilon}).
\end{align*}
Taking $\sigma_2$ such that
$\sigma_1-2\sigma_2=-\gamma_1-\sigma_1(1-\gamma_1)$,
that is, $\sigma_2=\{\gamma_1+\sigma_1(2-\gamma_1)\}/2$,
$$\int_{|t|\ge|v|^{-\sigma_1}}I\,dt=
O(|v|^{-\gamma_1-\sigma_1(1-\gamma_1)})+
O(|v|^{-2(\gamma_1-1)+\epsilon})$$
can be obtained. Since one can take $\sigma_1$ such that
$-\gamma_1-\sigma_1(1-\gamma_1)\le-(2\gamma_1-2)+\epsilon$,
that is, $\sigma_1\le(2+\epsilon-\gamma_1)/(\gamma_1-1)$,
we finally obtain
$$\int_{|t|\ge|v|^{-\sigma_1}}I\,dt=
O(|v|^{-2(\gamma_1-1)+\epsilon}),$$
which yields \eqref{2.17} in the proof of our Lemma \ref{lem2.4}.
\end{rem}

As in \cite{AM}, we introduce auxiliary wave operators
$$\Omega_{G,v}^{\mathrm{s},\pm}=\slim_{t\to\pm\infty}
e^{itH}U_{G,v}^\mathrm{s}(t),$$
where
$U_{G,v}^\mathrm{s}(t)=e^{-itH_0}M_{G,v}^\mathrm{s}(t)$
and $M_{G,v}^\mathrm{s}(t)=e^{-i\int_0^t V^\mathrm{s}(vs+e_1s^2/2)\,ds}$.
We know that 
$$\Omega_{G,v}^{\mathrm{s},\pm}=W^\pm I_{G,v}^{\mathrm{s},\pm},\quad
I_{G,v}^{\mathrm{s},\pm}=\slim_{t\to\pm\infty}M_{G,v}^\mathrm{s}(t)$$
exist, by virtue of \eqref{2.5}.
As emphasized in \cite{AM}, the $v$-dependent Graf-type modifier
$M_{G,v}^\mathrm{s}(t)$ commutes with any operators.
This fact will be used frequently.
Then the following can be obtained as in \cite{AM}, so we omit the proof:

\begin{lem}\label{lem2.5}
Let $v$ and $\varPhi_v$ be as in Theorem \ref{thm2.1},
and $\epsilon>0$. Then
\begin{equation}
\sup_{t\in\boldsymbol{R}}\|(e^{-itH}\Omega_{G,v}^{\mathrm{s},-}-U_{G,v}^\mathrm{s}(t))\varPhi_v\|
=O(|v|^{\max\{-1,-2(\gamma_1-1)+\epsilon\}})
\end{equation}
holds as $|v|\to\infty$ for $V^\mathrm{vs}\in\mathscr{V}^\mathrm{vs}$
and $V^\mathrm{s}\in\mathscr{V}^\mathrm{s}(1/2,1)$.
\end{lem}

Now we will show Theorem \ref{thm2.1}:

\begin{proof}[Proof of Theorem \ref{thm2.1}]
Since the proof is quite similar to the one of
Theorem 2.1 of \cite{AM}, we give its sketch only.

Suppose that $V^\mathrm{vs}\in\mathscr{V}^\mathrm{vs}$
and $V^\mathrm{s}\in\mathscr{V}^\mathrm{s}(1/2,5/4)$.
We first note that $S$ is represented as
\begin{align*}
&S=(W^+)^*W^-=I_{G,v}^\mathrm{s}(\Omega_{G,v}^{\mathrm{s},+})^*\Omega_{G,v}^{\mathrm{s},-},\\
&I_{G,v}^\mathrm{s}=I_{G,v}^{\mathrm{s},+}(I_{G,v}^{\mathrm{s},-})^*=e^{-i\int_{-\infty}^\infty V^\mathrm{s}(vs+e_1s^2/2)\,ds}.
\end{align*}
Noting $[S,p_j]=[S-I_{G,v}^\mathrm{s},p_j-v_j]$, $(p_j-v_j)\varPhi_v=(p_j\varPhi_0)_v$ and
\begin{align*}
i(S-I_{G,v}^\mathrm{s})\varPhi_v={}&I_{G,v}^\mathrm{s}i(\Omega_{G,v}^{\mathrm{s},+}-\Omega_{G,v}^{\mathrm{s},-})^*\Omega_{G,v}^{\mathrm{s},-}\varPhi_v\\
={}&I_{G,v}^\mathrm{s}\int_{-\infty}^\infty U_{G,v}^\mathrm{s}(t)^*V_te^{-itH}\Omega_{G,v}^{\mathrm{s},-}\varPhi_v\,dt
\end{align*}
with $V_t=V^\mathrm{vs}(x)+V^\mathrm{s}(x)-V^\mathrm{s}(vt+e_1t^2/2)$,
we have
$$|v|(i[S,p_j]\varPhi_v,\varPsi_v)=I_{G,v}^\mathrm{s}\{I(v)+R(v)\}$$
with
\begin{align*}
I(v)={}&|v|\int_{-\infty}^\infty
[(V_tU_{G,v}^\mathrm{s}(t)(p_j\varPhi_0)_v,U_{G,v}^\mathrm{s}(t)\varPsi_v)\\
{}&\qquad\qquad-(V_tU_{G,v}^\mathrm{s}(t)\varPhi_v,U_{G,v}^\mathrm{s}(t)(p_j\varPsi_0)_v)]\,dt,\\
R(v)={}&|v|\int_{-\infty}^\infty
[((e^{-itH}\Omega_{G,v}^{\mathrm{s},-}-U_{G,v}^\mathrm{s}(t))(p_j\varPhi_0)_v,V_tU_{G,v}^\mathrm{s}(t)\varPsi_v)\\
{}&\qquad\qquad-((e^{-itH}\Omega_{G,v}^{\mathrm{s},-}-U_{G,v}^\mathrm{s}(t))\varPhi_v,V_tU_{G,v}^\mathrm{s}(t)(p_j\varPsi_0)_v)]\,dt.
\end{align*}
By Lemmas \ref{lem2.3}, \ref{lem2.4} and \ref{lem2.5},
we have 
\begin{equation}
|R(v)|=O(|v|^{1+2\max\{-1,-2(\gamma_1-1)+\epsilon\}})=O(|v|^{\max\{-1,5-4\gamma_1+2\epsilon\}}).
\end{equation}
Then we need the condition $5-4\gamma_1<0$
in order to get $\lim_{|v|\to\infty}R(v)=0$, because
one can take $\epsilon>0$ so small
that $5-4\gamma_1+2\epsilon<0$.
This is equivalent to $\gamma_1>5/4$.

The rest of the proof is the same as in \cite{We1} and
\cite{AM}. So we omit it.
\end{proof}

By virtue of Theorem \ref{thm2.1} and the Plancherel formula associated with
the Radon transform (see Helgason~\cite{He}), Theorem \ref{thm1.1} can be shown
in the same way as in the proof of Theorem 1.2 of \cite{We1}
(see also \cite{EW}). Thus we omit its proof.

\section{The case where $V^\mathrm{s}\in\mathscr{V}^\mathrm{s}(1/2,5/4)$ and $V^\mathrm{l}\not=0$}

The main purpose of this section is showing the following
reconstruction formula, which yields Theorem \ref{thm1.2}.
For the sake of brevity, we put
\begin{equation}
U_D(t)=e^{-itH_0}M_D(t),\quad
M_D(t)=e^{-i\int_0^tV^\mathrm{l}(ps+e_1s^2/2)\,ds}.
\end{equation}

\begin{thm}\label{thm3.1}
Let $\hat{v}\in\boldsymbol{R}^n$ be given
such that $|\hat{v}\cdot e_1|<1$. Put $v=|v|\hat{v}$.
Let $\eta>0$ be given,
and $\varPhi_0,\,\varPsi_0\in L^2(\boldsymbol{R}^n)$ be
such that $\hat{\varPhi}_0,\,\hat{\varPsi}_0\in
C_0^\infty(\boldsymbol{R}^n)$ with
$\mathrm{supp}\,\hat{\varPhi}_0,\,
\mathrm{supp}\,\hat{\varPsi}_0\subset
\bigl\{\xi\in\boldsymbol{R}^n\bigm||\xi|<\eta\bigr\}$.
Put $\varPhi_v=e^{iv\cdot x}\varPhi_0$ and $\varPsi_v=e^{iv\cdot x}\varPsi_0$.
Let $V^\mathrm{vs}\in\mathscr{V}^\mathrm{vs}$,
$V^\mathrm{s}\in\mathscr{V}^\mathrm{s}(1/2,5/4)$,
and $V^\mathrm{l}\in\mathscr{V}_D^\mathrm{l}(3/8)$.
Then the following holds:
\begin{equation}
\begin{split}
&\lim_{|v|\to\infty}
\left\{|v|(i[S_D,p_j]\varPhi_v,\varPsi_v)-
\int_{-\infty}^\infty
i((\partial_jV^\mathrm{l})(x)U_D(t)\varPhi_v,U_D(t)\varPsi_v)\,dt\right\}\\
={}&\int_{-\infty}^\infty[
(V^\mathrm{vs}(x+\hat{v}\tau)p_j\varPhi_0,\varPsi_0)-
(V^\mathrm{vs}(x+\hat{v}\tau)\varPhi_0,p_j\varPsi_0)\\
&\qquad\qquad+i((\partial_j V^\mathrm{s})(x+\hat{v}\tau)\varPhi_0,\varPsi_0)]
\,d\tau
\end{split}\label{3.2}
\end{equation}
for $1\le j\le n$.
\end{thm}

We first need the following lemma:

\begin{lem}\label{lem3.2}
Let $v$ and $\varPhi_v$ be as in Theorem \ref{thm3.1},
and $V^\mathrm{l}\in\mathscr{V}_D^\mathrm{l}(1/4)$ with $\gamma_D<1/2$.
Then, for $0\le\nu_1,\,\nu_2,\,\nu_3\le1$,
there exists a positive constant $C$ such that
\begin{equation}
\begin{split}
&\|\langle x\rangle^2M_D(t)\varPhi_v\|
=\|\langle x\rangle^2M_{D,v}(t)\varPhi_0\|\\
\le{}&
C(1+|v|^{-(2\gamma_D+1)\nu_1}|t|^{4-(2\gamma_D+1)(2-\nu_1)}
+|v|^{-(\gamma_D+1)\nu_2}|t|^{3-(\gamma_D+1)(2-\nu_2)}\\
&\qquad\qquad
+|v|^{-(\gamma_D+1/2)\nu_3}|t|^{2-(\gamma_D+1/2)(2-\nu_3)})
\end{split}
\label{3.3}
\end{equation}
holds as $|v|\to\infty$, where
$M_{D,v}(t)=e^{-iv\cdot x}M_D(t)e^{iv\cdot x}=e^{-i\int_0^tV^\mathrm{l}(ps+vs+e_1s^2/2)\,ds}$.
\end{lem}

\begin{proof}
First of all, we introduce
$$\tilde{V}_{|v|,t}^\mathrm{l}(x)=V^\mathrm{l}(x+vt+e_1t^2/2)
g((\lambda_1|v|\langle t\rangle)^{-1}x),$$
by mimicking the definition of $\tilde{V}_{|v|,t}^\mathrm{s}(x)$
in the proof of Lemma \ref{lem2.4}.
Since $\mathrm{supp}\,\hat{\varPhi}_0\subset
\bigl\{\xi\in\boldsymbol{R}^n\bigm||\xi|<\eta\bigr\}$,
\begin{equation}
M_{D,v}(t)\varPhi_0
=e^{-i\int_0^t\tilde{V}_{|v|,s}^\mathrm{l}(ps)\,ds}\varPhi_0
\label{3.4}
\end{equation}
holds for $|v|\ge\eta/(3\lambda_1)$.
Since $x=i\nabla_p$,
\begin{equation}
e^{i\int_0^t\tilde{V}_{|v|,s}^\mathrm{l}(ps)\,ds}xe^{-i\int_0^t\tilde{V}_{|v|,s}^\mathrm{l}(ps)\,ds}=
x+\int_0^ts(\nabla \tilde{V}_{|v|,s}^\mathrm{l})(ps)\,ds
\label{3.5}
\end{equation}
holds. This yields
\begin{align*}
&e^{i\int_0^t\tilde{V}_{|v|,s}^\mathrm{l}(ps)\,ds}\langle x\rangle^2e^{-i\int_0^t\tilde{V}_{|v|,s}^\mathrm{l}(ps)\,ds}\\
={}&1+\left(x+\int_0^ts(\nabla \tilde{V}_{|v|,s}^\mathrm{l})(ps)\,ds\right)^2\\
={}&\langle x\rangle^2+i\left(\int_0^ts^2(\Delta \tilde{V}_{|v|,s}^\mathrm{l})(ps)\,ds\right)
+2\left(\int_0^ts(\nabla \tilde{V}_{|v|,s}^\mathrm{l})(ps)\,ds\right)\cdot x\\
&\qquad+\left(\int_0^ts(\nabla \tilde{V}_{|v|,s}^\mathrm{l})(ps)\,ds\right)^2.
\end{align*}
Now we will estimate $\|\nabla\tilde{V}_{|v|,t}^\mathrm{l}\|_{L^\infty}$
and $\|\Delta\tilde{V}_{|v|,t}^\mathrm{l}\|_{L^\infty}$.
In the same way as in the proof of Lemma \ref{lem2.4},
for $0\le\nu_1,\,\nu_2,\,\nu_3,\,\nu_4,\,\nu_5\le1$ and $|v|\ge1$,
\begin{equation}
\begin{split}
\|\nabla\tilde{V}_{|v|,t}^\mathrm{l}\|_{L^\infty}
\le{}&
C_1(1+|v|^{\nu_1/(2-\nu_1)}|t|)^{-(\gamma_D+1/2)(2-\nu_1)}\\
&+C_2|v|^{\nu_2/(2-\nu_2)-1}(1+|v|^{\nu_2/(2-\nu_2)}|t|)^{-\gamma_D(2-\nu_2)-1}
\end{split}\label{3.6}
\end{equation}
and
\begin{equation}
\begin{split}
\|\Delta\tilde{V}_{|v|,t}^\mathrm{l}\|_{L^\infty}
\le{}&
C_3(1+|v|^{\nu_3/(2-\nu_3)}|t|)^{-(\gamma_D+1)(2-\nu_3)}\\
&+C_4|v|^{\nu_4/(2-\nu_4)-1}(1+|v|^{\nu_4/(2-\nu_4)}|t|)^{-(\gamma_D+1/2)(2-\nu_4)-1}\\
&\quad+C_5|v|^{2\nu_5/(2-\nu_5)-2}(1+|v|^{\nu_5/(2-\nu_5)}|t|)^{-\gamma_D(2-\nu_5)-2}
\end{split}\label{3.7}
\end{equation}
can be obtained. Noting $1-(\gamma_D+1/2)(2-\nu_1)\ge-2\gamma_D>-1$
and $1-\gamma_D(2-\nu_2)-1\ge-2\gamma_D>-1$,
\begin{align*}
&\left\|\int_0^ts(\nabla \tilde{V}_{|v|,s}^\mathrm{l})(ps)\,ds\right\|_{\mathscr{B}(L^2)}\\
\le{}&C_1\int_0^{|t|}s(|v|^{\nu_1/(2-\nu_1)}s)^{-(\gamma_D+1/2)(2-\nu_1)}\,ds\\
&+C_2|v|^{\nu_2/(2-\nu_2)-1}\int_0^{|t|}s(|v|^{\nu_2/(2-\nu_2)}s)^{-\gamma_D(2-\nu_2)-1}\,ds\\
={}&C_1'|v|^{-(\gamma_D+1/2)\nu_1}|t|^{2-(\gamma_D+1/2)(2-\nu_1)}
+C_2'|v|^{-\gamma_D\nu_2-1}|t|^{1-\gamma_D(2-\nu_2)}\\
={}&C_1'|v|^{1-2\gamma_D}(|v|^{-1}|t|)^{2-(\gamma_D+1/2)(2-\nu_1)}
+C_2'|v|^{-2\gamma_D}(|v|^{-1}|t|)^{1-\gamma_D(2-\nu_2)}
\end{align*}
can be obtained. Taking account of
\begin{align*}
&[1-2\gamma_D,3/2-\gamma_D]=\bigl\{2-(\gamma_D+1/2)(2-\nu_1)\bigm|0\le\nu_1\le1\bigr\}\\
&\qquad\supset[1-2\gamma_D,1-\gamma_D]=\bigl\{1-\gamma_D(2-\nu_2)\bigm|0\le\nu_2\le1\bigr\},
\end{align*}
we obtain the estimate
\begin{equation}
\left\|\int_0^ts(\nabla \tilde{V}_{|v|,s}^\mathrm{l})(ps)\,ds\right\|_{\mathscr{B}(L^2)}\le
C_1''|v|^{-(\gamma_D+1/2)\nu_1}|t|^{2-(\gamma_D+1/2)(2-\nu_1)}
\label{3.8}
\end{equation}
for $0\le\nu_1\le1$. Noting $2-(\gamma_D+1)(2-\nu_3)\ge-2\gamma_D>-1$,
$2-(\gamma_D+1/2)(2-\nu_4)-1\ge-2\gamma_D>-1$ and
$2-\gamma_D(2-\nu_5)-2\ge-2\gamma_D>-1$,
\begin{align*}
&\left\|\int_0^ts^2(\Delta \tilde{V}_{|v|,s}^\mathrm{l})(ps)\,ds\right\|_{\mathscr{B}(L^2)}\\
\le{}&C_3'|v|^{-(\gamma_D+1)\nu_3}|t|^{3-(\gamma_D+1)(2-\nu_3)}
+C_4'|v|^{-(\gamma_D+1/2)\nu_4-1}|t|^{2-(\gamma_D+1/2)(2-\nu_4)}\\
&+C_5'|v|^{-\gamma_D\nu_5-2}|t|^{1-\gamma_D(2-\nu_5)}\\
={}&C_3'|v|^{1-2\gamma_D}(|v|^{-1}|t|)^{3-(\gamma_D+1)(2-\nu_3)}
+C_4'|v|^{-2\gamma_D}(|v|^{-1}|t|)^{2-(\gamma_D+1/2)(2-\nu_4)}\\
&+C_5'|v|^{-1-2\gamma_D}(|v|^{-1}|t|)^{1-\gamma_D(2-\nu_5)}
\end{align*}
can be obtained. Taking account of
\begin{align*}
&[1-2\gamma_D,2-\gamma_D]=\bigl\{3-(\gamma_D+1)(2-\nu_3)\bigm|0\le\nu_3\le1\bigr\}\\
&\qquad\supset[1-2\gamma_D,3/2-\gamma_D]=\bigl\{2-(\gamma_D+1/2)(2-\nu_4)\bigm|0\le\nu_4\le1\bigr\}\\
&\qquad\qquad\supset[1-2\gamma_D,1-\gamma_D]=\bigl\{1-\gamma_D(2-\nu_5)\bigm|0\le\nu_5\le1\bigr\},
\end{align*}
we obtain the estimate
\begin{equation}
\left\|\int_0^ts^2(\Delta \tilde{V}_{|v|,s}^\mathrm{l})(ps)\,ds\right\|_{\mathscr{B}(L^2)}\le
C_3''|v|^{-(\gamma_D+1)\nu_3}|t|^{3-(\gamma_D+1)(2-\nu_3)}
\label{3.9}
\end{equation}
for $0\le\nu_3\le1$.

Based on the above observations, the lemma can be proved.
\end{proof}

Then the following lemma can be shown as in \cite{AM}.

\begin{lem}\label{lem3.3}
Let $v$ and $\varPhi_v$ be as in Theorem \ref{thm3.1},
and $V^\mathrm{l}\in\mathscr{V}_D^\mathrm{l}(1/4)$. Then
\begin{equation}
\int_{-\infty}^\infty
\|V^\mathrm{vs}(x)U_D(t)\varPhi_v\|\,dt
=O(|v|^{-1})
\label{3.10}
\end{equation}
holds as $|v|\to\infty$ for $V^\mathrm{vs}\in\mathscr{V}^\mathrm{vs}$.
\end{lem}

\begin{proof}
For the sake of brevity, we put $I=\|V^\mathrm{vs}(x)U_D(t)\varPhi_v\|$.
For simplicity, we suppose $\gamma_D<1/2$.
Since by virtue of the Avron-Herbst formula \eqref{2.8},
$I$ can be written as
\begin{align*}
I={}&\|V^\mathrm{vs}(x+vt+e_1t^2/2)e^{-itK_0}M_{D,v}(t)\varPhi_0\|\\
={}&\|V^\mathrm{vs}(x+vt+e_1t^2/2)(1+K_0)^{-1}e^{-itK_0}f(p)M_{D,v}(t)(1+K_0)\varPhi_0\|,
\end{align*}
we estimate this in the same way as in the proof of
Lemma \ref{lem2.3}, which is omitted in this paper (cf. the proof of
Lemma \ref{lem2.4}).
$$\|V^\mathrm{vs}(x)U_D(t)\varPhi_v\|\le
I_1+I_2+I_3,$$
where
\begin{align*}
I_1={}&
\|V^\mathrm{vs}(x+vt+e_1t^2/2)(1+K_0)^{-1}\\
&\times F(|x|\ge3\lambda_1|v||t|)e^{-itK_0}f(p)F(|x|\le\lambda_1|v||t|)M_{D,v}(t)(1+K_0)\varPhi_0\|,\\
I_2={}&
\|V^\mathrm{vs}(x+vt+e_1t^2/2)(1+K_0)^{-1}\\
&\times F(|x|\ge3\lambda_1|v||t|)e^{-itK_0}f(p)F(|x|>\lambda_1|v||t|)M_{D,v}(t)(1+K_0)\varPhi_0\|\\
I_3={}&
\|V^\mathrm{vs}(x+vt+e_1t^2/2)(1+K_0)^{-1}\\
&\times F(|x|<3\lambda_1|v||t|)e^{-itK_0}f(p)M_{D,v}(t)(1+K_0)\varPhi_0\|.
\end{align*}
As for $I_1$,
by virtue of Proposition \ref{prop2.2}
and $\|V^\mathrm{vs}(x+vt+e_1t^2/2)(1+K_0)^{-1}\|_{\mathscr{B}(L^2)}=
\|V^\mathrm{vs}(1+K_0)^{-1}\|_{\mathscr{B}(L^2)}$,
$$I_1\le C(1+\lambda_1|v||t|)^{-2}$$
holds for $|v|>\eta/\lambda_1$,
which yields
$$\int_{-\infty}^\infty I_1\,dt=O(|v|^{-1}).$$
As for $I_3$, by that \eqref{2.12} holds
for $|t|\ge1$ and
$|x|<3\lambda_1|v||t|<4\lambda_1|v|\langle t\rangle$, and \eqref{1.8},
$$\int_{-\infty}^\infty I_3\,dt=O(|v|^{-1})$$
can be obtained.
As for $I_2$, by virtue of Lemma \ref{lem3.2},
\begin{align*}
I_2\le{}&
C(1+\lambda_1|v||t|)^{-2}(1+|v|^{-(2\gamma_D+1)\nu_1}|t|^{4-(2\gamma_D+1)(2-\nu_1)}\\
&\qquad+|v|^{-(\gamma_D+1)\nu_2}|t|^{3-(\gamma_D+1)(2-\nu_2)}
+|v|^{-(\gamma_D+1/2)\nu_3}|t|^{2-(\gamma_D+1/2)(2-\nu_3)})
\end{align*}
holds for $0\le\nu_1,\,\nu_2,\,\nu_3\le1$
(cf. $I_2$ in the proof of Lemma \ref{lem2.3}).
Therefore,
\begin{align*}
\int_{-\infty}^\infty I_2\,dt={}&O(|v|^{-1})+
O(|v|^{-(2\gamma_D+1)\nu_1-\{5-(2\gamma_D+1)(2-\nu_1)\}})\\
&+O(|v|^{-(\gamma_D+1)\nu_2-\{4-(\gamma_D+1)(2-\nu_2)\}})\\
&\quad+
O(|v|^{-(\gamma_D+1/2)\nu_3-\{3-(\gamma_D+1/2)(2-\nu_3)\}})\\
={}&O(|v|^{-1})+
O(|v|^{-3+4\gamma_D-2(2\gamma_D+1)\nu_1})+O(|v|^{-2+2\gamma_D-2(\gamma_D+1)\nu_2})\\
&\quad+
O(|v|^{-2+2\gamma_D-2(\gamma_D+1/2)\nu_3})
\end{align*}
can be obtained by an appropriate change of variables
under the conditions
$-2+4-(2\gamma_D+1)(2-\nu_1)<-1$,
$-2+3-(\gamma_D+1)(2-\nu_2)<-1$ and
$-2+2-(\gamma_D+1/2)(2-\nu_3)<-1$, i.e.
$\nu_1<2-3/(2\gamma_D+1)$,
$\nu_2<2-2/(\gamma_D+1)$
and $\nu_3<2-1/(\gamma_D+1/2)$.
Since $0<2-3/(2\gamma_D+1)<1/2$,
$2/5<2-2/(\gamma_D+1)<2/3$,
$2/3<2-1/(\gamma_D+1/2)<1$, and
\begin{align*}
&\inf_{0\le\nu_1<2-3/(2\gamma_D+1)}(-3+4\gamma_D-2(2\gamma_D+1)\nu_1)=-1-4\gamma_D,\\
&\inf_{0\le\nu_2<2-2/(\gamma_D+1)}(-2+2\gamma_D-2(\gamma_D+1)\nu_2)=-2-2\gamma_D,\\
&\inf_{0\le\nu_3<2-1/(\gamma_D+1/2)}(-2+2\gamma_D-2(\gamma_D+1/2)\nu_3)=-2-2\gamma_D,
\end{align*}
we obtain
$$\int_{-\infty}^\infty I_2\,dt=O(|v|^{-1}).$$

Based on the above observations, the
lemma can be proved.
\end{proof}

The following lemma can be also proved in the same way as in the
proof of Lemmas \ref{lem2.4} and \ref{lem3.3}.

\begin{lem}\label{lem3.4}
Let $v$ and $\varPhi_v$ be as in Theorem \ref{thm3.1},
$\epsilon>0$, and $V^\mathrm{l}\in\mathscr{V}_D^\mathrm{l}(1/4)$.
Then
\begin{equation}
\int_{-\infty}^\infty
\|\{V^\mathrm{s}(x)-V^\mathrm{s}(vt+e_1t^2/2)\}U_D(t)\varPhi_v\|\,dt
=O(|v|^{\max\{-1,-2(\gamma_1-1)+\epsilon\}})
\label{3.11}
\end{equation}
holds as $|v|\to\infty$ for $V^\mathrm{s}\in\mathscr{V}^\mathrm{s}(1/2,1)$.
\end{lem}

\begin{proof}
Since the proof is quite similar to the proof of
Lemma \ref{lem2.4}, we sketch it:
For the sake of brevity, we put $I=\|\{V^\mathrm{s}(x)-V^\mathrm{s}(vt+e_1t^2/2)\}U_D(t)\varPhi_v\|$.
For simplicity, we suppose $\gamma_D<1/2$.
As in the proof of Lemma \ref{lem2.4}, we estimate $I$ as
\begin{equation}
\begin{split}
I={}&\|\{V^\mathrm{s}(x+vt+e_1t^2/2)-V^\mathrm{s}(vt+e_1t^2/2)\}e^{-itK_0}
M_{D,v}(t)\varPhi_0\|\\
\le{}&\|\bar{V}_{|v|,t}^\mathrm{s}(x)e^{-itK_0}
f(p)M_{D,v}(t)\varPhi_0\|\\
&+
\|\{\tilde{V}_{|v|,t}^\mathrm{s}(x)-\tilde{V}_{|v|,t}^\mathrm{s}(0)\}e^{-itK_0}
M_{D,v}(t)\varPhi_0\|.
\end{split}\label{3.12}
\end{equation}
Here we used $M_{D,v}(t)\varPhi_0=f(p)M_{D,v}(t)\varPhi_0$.
As for the first term in the inequality \eqref{3.12},
we estimate it as
$$\|\bar{V}_{|v|,t}^\mathrm{s}(x)e^{-itK_0}
f(p)M_{D,v}(t)\varPhi_0\|\le
I_1+I_2+I_3,$$
where
\begin{align*}
I_1={}&
\|\bar{V}_{|v|,t}^\mathrm{s}(x)F(|x|\ge3\lambda_1|v||t|)e^{-itK_0}f(p)F(|x|\le\lambda_1|v||t|)
M_{D,v}(t)\varPhi_0\|,\\
I_2={}&
\|\bar{V}_{|v|,t}^\mathrm{s}(x)F(|x|\ge3\lambda_1|v||t|)e^{-itK_0}f(p)F(|x|>\lambda_1|v||t|)
M_{D,v}(t)\varPhi_0\|,\\
I_3={}&
\|\bar{V}_{|v|,t}^\mathrm{s}(x)F(|x|<3\lambda_1|v||t|)e^{-itK_0}f(p)M_{D,v}(t)\varPhi_0\|.
\end{align*}
As for $I_1$,
by virtue of Proposition \ref{prop2.2},
$$I_1\le C(1+\lambda_1|v||t|)^{-2}$$
holds for $|v|>\eta/\lambda_1$.
As for $I_2$, by virtue of Lemma \ref{lem3.2},
\begin{align*}
I_2\le{}&
C(1+\lambda_1|v||t|)^{-2}(1+|v|^{-(2\gamma_D+1)\nu_1}|t|^{4-(2\gamma_D+1)(2-\nu_1)}\\
&\qquad+|v|^{-(\gamma_D+1)\nu_2}|t|^{3-(\gamma_D+1)(2-\nu_2)}
+|v|^{-(\gamma_D+1/2)\nu_3}|t|^{2-(\gamma_D+1/2)(2-\nu_3)})
\end{align*}
holds. As for $I_3$, by the definition of
$\tilde{V}_{|v|,t}^\mathrm{s}(x)$, $I_3=0$ holds.
Combining these estimates with the results in the proof of Lemma \ref{lem3.3},
we obtain
\begin{equation}
\int_{-\infty}^\infty\|\bar{V}_{|v|,t}^\mathrm{s}(x)e^{-itK_0}
f(p)M_{D,v}(t)\varPhi_0\|\,dt=O(|v|^{-1}).
\label{3.13}
\end{equation}
As for the second term of the inequality \eqref{3.12},
in the same way as in the proof
of Lemma \ref{lem2.4}, we see that
\begin{align*}
&\|\{\tilde{V}_{|v|,t}^\mathrm{s}(x)-\tilde{V}_{|v|,t}^\mathrm{s}(0)\}e^{-itK_0}
M_{D,v}(t)\varPhi_0\|\\
\le{}&\|\nabla\tilde{V}_{|v|,t}^\mathrm{s}\|_{L^\infty}\|(x+pt)M_{D,v}(t)\varPhi_0\|\\
\le{}&\|\nabla\tilde{V}_{|v|,t}^\mathrm{s}\|_{L^\infty}\left(\|x\varPhi_0\|+\left\|\left(\int_0^ts(\nabla \tilde{V}_{|v|,s}^\mathrm{l})(ps)\,ds\right)\varPhi_0\right\|+|t|\|p\varPhi_0\|\right)
\end{align*}
holds.
Here we used \eqref{3.4} and \eqref{3.5}.
As for the estimates of
$$\int_{-\infty}^\infty\|\nabla\tilde{V}_{|v|,t}^\mathrm{s}\|_{L^\infty}
\|x\varPhi_0\|\,dt,\quad
\int_{-\infty}^\infty\|\nabla\tilde{V}_{|v|,t}^\mathrm{s}\|_{L^\infty}
|t|\|p\varPhi_0\|\,dt,$$
we already obtained \eqref{2.15} and \eqref{2.16} in the proof of
Lemma \ref{lem2.4}. Hence, we have only to estimate
$$\int_{-\infty}^\infty\|\nabla\tilde{V}_{|v|,t}^\mathrm{s}\|_{L^\infty}
\left\|\left(\int_0^ts(\nabla \tilde{V}_{|v|,s}^\mathrm{l})(ps)\,ds\right)\varPhi_0\right\|\,dt.$$
To this end, we consider
\begin{align*}
\tilde{I}_1
={}&\int_{-\infty}^\infty
(1+|v|^{\nu_1/(2-\nu_1)}|t|)^{-\gamma_1(2-\nu_1)}|v|^{-(\gamma_D+1/2)\nu_3}|t|^{2-(\gamma_D+1/2)(2-\nu_3)}\,dt\\
={}&O(|v|^{-(\gamma_D+1/2)\nu_3-\{3-(\gamma_D+1/2)(2-\nu_3)\}\nu_1/(2-\nu_1)}),\\
\tilde{I}_2
={}&\int_{-\infty}^\infty|v|^{\nu_2/(2-\nu_2)-1}(1+|v|^{\nu_2/(2-\nu_2)}|t|)^{-\gamma_0(2-\nu_2)-1}\\
&\qquad\qquad\qquad\times|v|^{-(\gamma_D+1/2)\nu_4}|t|^{2-(\gamma_D+1/2)(2-\nu_4)}\,dt\\
={}&O(|v|^{\nu_2/(2-\nu_2)-1-(\gamma_D+1/2)\nu_4-\{3-(\gamma_D+1/2)(2-\nu_4)\}\nu_2/(2-\nu_2)})\\
={}&O(|v|^{-1-(\gamma_D+1/2)\nu_4-\{2-(\gamma_D+1/2)(2-\nu_4)\}\nu_2/(2-\nu_2)}),
\end{align*}
which can be obtained by an appropriate change of variables
under the conditions 
$-\gamma_1(2-\nu_1)+2-(\gamma_D+1/2)(2-\nu_3)<-1$ and
$-\gamma_0(2-\nu_2)-1+2-(\gamma_D+1/2)(2-\nu_4)<-1$,
i.e. $\gamma_1\nu_1+(\gamma_D+1/2)\nu_3<2\{(\gamma_1-1)+\gamma_D\}$
and $\gamma_0\nu_2+(\gamma_D+1/2)\nu_4<2(\gamma_0+\gamma_D)-1$.
Noting
$$\gamma_1+(\gamma_D+1/2)-2\{(\gamma_1-1)+\gamma_D\}
=5/2-\gamma_1-\gamma_D
>3/2-\gamma_0-1/2\ge0$$
by $\gamma_0\le1$, $\gamma_1\le1+\gamma_0$ and $\gamma_D<1/2$, we see
$$\bigl\{\gamma_1\nu_1+(\gamma_D+1/2)\nu_3\bigm|0\le\nu_1,\,\nu_3\le1\bigr\}
\supset[0,2\{(\gamma_1-1)+\gamma_D\}].$$
We also note that for $\nu_1$ and $\nu_3$ such that
$\gamma_1\nu_1+(\gamma_D+1/2)\nu_3=2\{(\gamma_1-1)+\gamma_D\}$,
\begin{align*}
&-(\gamma_D+1/2)\nu_3-\{3-(\gamma_D+1/2)(2-\nu_3)\}\nu_1/(2-\nu_1)\\
={}&-(\gamma_D+1/2)\nu_3-\{2-2\gamma_D+(\gamma_D+1/2)\nu_3\}\\
&\qquad\qquad\qquad\qquad\qquad\times\frac{2\{(\gamma_1-1)+\gamma_D\}-(\gamma_D+1/2)\nu_3}{2\gamma_1-2\{(\gamma_1-1)+\gamma_D\}+(\gamma_D+1/2)\nu_3}\\
={}&-2\{(\gamma_1-1)+\gamma_D\}
\end{align*}
holds. This yields
\begin{equation}
\begin{split}
&\inf_{\nu_1,\,\nu_3}
(-(\gamma_D+1/2)\nu_3-\{3-(\gamma_D+1/2)(2-\nu_3)\}\nu_1/(2-\nu_1))\\
={}&-2\{(\gamma_1-1)+\gamma_D\}.
\end{split}\label{3.14}
\end{equation}
Noting
$$\gamma_0+(\gamma_D+1/2)-\{2(\gamma_0+\gamma_D)-1\}
=3/2-\gamma_0-\gamma_D>0$$
by $\gamma_0\le1$ and $\gamma_D<1/2$, we see
$$\bigl\{\gamma_0\nu_2+(\gamma_D+1/2)\nu_4\bigm|0\le\nu_2,\,\nu_4\le1\bigr\}
\supset[0,2(\gamma_0+\gamma_D)-1].$$
We also note that for $\nu_2$ and $\nu_4$ such that
$\gamma_0\nu_2+(\gamma_D+1/2)\nu_4=2(\gamma_0+\gamma_D)-1$,
\begin{align*}
&-1-(\gamma_D+1/2)\nu_4-\{2-(\gamma_D+1/2)(2-\nu_4)\}\nu_2/(2-\nu_2)\\
={}&-1-(\gamma_D+1/2)\nu_4-\{1-2\gamma_D+(\gamma_D+1/2)\nu_4\}\\
&\qquad\qquad\qquad\qquad\qquad\qquad\times\frac{\{2(\gamma_0+\gamma_D)-1\}-(\gamma_D+1/2)\nu_4}{2\gamma_0-\{2(\gamma_0+\gamma_D)-1\}+(\gamma_D+1/2)\nu_4}\\
={}&-1-\{2(\gamma_0+\gamma_D)-1\}=-2(\gamma_0+\gamma_D)
\end{align*}
holds. This yields
\begin{equation}
\begin{split}
&\inf_{\nu_2,\,\nu_4}
(-1-(\gamma_D+1/2)\nu_4-\{2-(\gamma_D+1/2)(2-\nu_4)\}\nu_2/(2-\nu_2))\\
={}&-2(\gamma_0+\gamma_D).
\end{split}\label{3.15}
\end{equation}
\eqref{3.14}, \eqref{3.15} and $-2(\gamma_0+\gamma_D)\le-2\{(\gamma_1-1)+\gamma_D\}$ imply
\begin{equation}
\int_{-\infty}^\infty\|\nabla\tilde{V}_{|v|,t}^\mathrm{s}\|_{L^\infty}
\left\|\left(\int_0^ts(\nabla \tilde{V}_{|v|,s}^\mathrm{l})(ps)\,ds\right)\varPhi_0\right\|\,dt=O(|v|^{-2\{(\gamma_1-1)+\gamma_D\}+\epsilon})
\label{3.16}
\end{equation}
with $\epsilon>0$.

By \eqref{3.13}, \eqref{2.15}, \eqref{2.16} and \eqref{3.16}, we obtain
\eqref{3.11}.
\end{proof}

The following lemma is the key in this section.

\begin{lem}\label{lem3.5} Let $v$ and $\varPhi_v$ be as in Theorem \ref{thm3.1},
$\epsilon>0$, and
$V^\mathrm{l}\in\mathscr{V}_D^\mathrm{l}(1/4)$.
Then
\begin{equation}
\int_{-\infty}^\infty
\|\{V^\mathrm{l}(x)-V^\mathrm{l}(pt-e_1t^2/2)\}U_D(t)\varPhi_v\|\,dt
=O(|v|^{-(4\gamma_D-1)+\epsilon})
\label{3.17}
\end{equation}
holds as $|v|\to\infty$.
\end{lem}

\begin{proof}
For the sake of brevity, we put $I=
\|\{V^\mathrm{l}(x)-V^\mathrm{l}(pt-e_1t^2/2)\}U_D(t)\varPhi_v\|$.
For simplicity, we suppose $\gamma_D<1/2$.
We first note that by virtue of the Avron-Herbst formula \eqref{2.8},
\begin{equation}
\begin{split}
&e^{-itH_0}i\frac{d}{dt}(M_D(t))=
e^{-itH_0}V^\mathrm{l}(pt+e_1t^2/2)M_D(t)\\
={}&V^\mathrm{l}((p-e_1t)t+e_1t^2/2)e^{-itH_0}M_D(t)
=V^\mathrm{l}(pt-e_1t^2/2)U_D(t)
\label{3.18}
\end{split}
\end{equation}
holds. In the same way as in the proof of Lemma \ref{lem3.4},
$I$ can be written as
\begin{align*}
I={}&\|\{V^\mathrm{l}(x+e_1t^2/2)-V^\mathrm{l}(pt+e_1t^2/2)\}e^{-itK_0}M_D(t)\varPhi_v\|\\
={}&\|\{V^\mathrm{l}(x+vt+e_1t^2/2)-V^\mathrm{l}(pt+vt+e_1t^2/2)\}e^{-itK_0}M_{D,v}(t)\varPhi_0\|.
\end{align*}
Now we will deal with this by using $\tilde{V}_{|v|,t}^\mathrm{l}(x)$ which is
introduced in the proof of Lemma \ref{lem3.2},
and mimicking the argument in the proof of Lemma \ref{lem3.4}.
Hence, we estimate it as
\begin{equation}
\begin{split}
I={}&\|\{V^\mathrm{l}(x+vt+e_1t^2/2)-\tilde{V}_{|v|,t}^\mathrm{l}(pt)\}e^{-itK_0}M_{D,v}(t)\varPhi_0\|\\
\le{}&\|\bar{V}_{|v|,t}^\mathrm{l}(x)e^{-itK_0}
f(p)M_{D,v}(t)\varPhi_0\|\\
&+
\|\{\tilde{V}_{|v|,t}^\mathrm{l}(x)-\tilde{V}_{|v|,t}^\mathrm{l}(pt)\}e^{-itK_0}
M_{D,v}(t)\varPhi_0\|
\end{split}\label{3.19}
\end{equation}
for $|v|\ge\eta/(3\lambda_1)$, where
$\bar{V}_{|v|,t}^\mathrm{l}(x)=V^\mathrm{l}(x+vt+e_1t^2/2)-\tilde{V}_{|v|,t}^\mathrm{l}(x)$.
As for the first term of the inequality \eqref{3.19},
in the same way as in the proof of Lemma \ref{lem3.4},
\begin{equation}
\int_{-\infty}^\infty\|\bar{V}_{|v|,t}^\mathrm{l}(x)e^{-itK_0}
f(p)M_{D,v}(t)\varPhi_0\|\,dt=O(|v|^{-1})
\label{3.20}
\end{equation}
can be obtained, because $I_3=0$ also in the long-range case.
As for the second term of the inequality \eqref{3.19},
we first note that by virtue of the Baker-Campbell-Hausdorff formula,
$\tilde{V}_{|v|,t}^\mathrm{l}(x)-\tilde{V}_{|v|,t}^\mathrm{l}(pt)$
can be written as
\begin{equation}
\begin{split}
\tilde{V}_{|v|,t}^\mathrm{l}(x)-\tilde{V}_{|v|,t}^\mathrm{l}(pt)
={}&\left(\int_0^1(\nabla \tilde{V}_{|v|,t}^\mathrm{l})(\theta x+(1-\theta)pt)\,d\theta\right)\cdot(x-pt)\\
&\quad+\frac{i}{2}\int_0^1t(\Delta \tilde{V}_{|v|,t}^\mathrm{l})(\theta x+(1-\theta)pt)\,d\theta.
\end{split}
\end{equation}
By virtue of this, we see that
\begin{align*}
&\|\{\tilde{V}_{|v|,t}^\mathrm{l}(x)-\tilde{V}_{|v|,t}^\mathrm{l}(pt)\}e^{-itK_0}
M_{D,v}(t)\varPhi_0\|\\
\le{}&\left\|\int_0^1(\nabla \tilde{V}_{|v|,t}^\mathrm{l})(\theta x+(1-\theta)pt)\,d\theta\right\|_{\mathscr{B}(L^2)}\|(x-pt)
e^{-itK_0}M_{D,v}(t)\varPhi_0\|\\
&+\left\|\frac{i}{2}\int_0^1t(\Delta \tilde{V}_{|v|,t}^\mathrm{l})(\theta x+(1-\theta)pt)\,d\theta\right\|_{\mathscr{B}(L^2)}
\|e^{-itK_0}M_{D,v}(t)\varPhi_0\|\\
\le{}&\|\nabla\tilde{V}_{|v|,t}^\mathrm{l}\|_{L^\infty}\|xM_{D,v}(t)\varPhi_0\|+\frac{1}{2}|t|\|\Delta\tilde{V}_{|v|,t}^\mathrm{l}\|_{L^\infty}\|\varPhi_0\|\\
\le{}&\|\nabla\tilde{V}_{|v|,t}^\mathrm{l}\|_{L^\infty}\left(\|x\varPhi_0\|+\left\|\left(\int_0^ts(\nabla \tilde{V}_{|v|,s}^\mathrm{l})(ps)\,ds\right)\varPhi_0\right\|\right)\\
&+\frac{1}{2}|t|\|\Delta\tilde{V}_{|v|,t}^\mathrm{l}\|_{L^\infty}\|\varPhi_0\|
\end{align*}
holds.
Here we used $e^{itK_0}(x-pt)e^{-itK_0}=x$,
\eqref{3.4} and \eqref{3.5}. In order to estimate
$$\int_{-\infty}^\infty\|\nabla\tilde{V}_{|v|,t}^\mathrm{l}\|_{L^\infty}
\|x\varPhi_0\|\,dt,$$
we consider 
\begin{align*}
\tilde{I}_{1,1}={}&\int_{-\infty}^\infty(1+|v|^{\nu_1/(2-\nu_1)}|t|)^{-(\gamma_D+1/2)(2-\nu_1)}\,dt=O(|v|^{-\nu_1/(2-\nu_1)}),\\
\tilde{I}_{1,2}={}&\int_{-\infty}^\infty|v|^{\nu_2/(2-\nu_2)-1}(1+|v|^{\nu_2/(2-\nu_2)}|t|)^{-\gamma_D(2-\nu_2)-1}\,dt\\
={}&O(|v|^{\nu_2/(2-\nu_2)-1-\nu_2/(2-\nu_2)})=O(|v|^{-1})
\end{align*}
by \eqref{3.6}, which can be obtained by an appropriate change of variables
under the conditions 
$-(\gamma_D+1/2)(2-\nu_1)<-1$ and
$-\gamma_D(2-\nu_2)-1<-1$, i.e. $\nu_1<2-1/(\gamma_D+1/2)$
and $\nu_2<2$.
Since $2/3<2-1/(\gamma_D+1/2)<1$,
$$\inf_{0\le\nu_1<2-1/(\gamma_D+1/2)}(-\nu_1/(2-\nu_1))=-2\gamma_D,$$
and $-2\gamma_D>-1$, we obtain
\begin{equation}
\int_{-\infty}^\infty\|\nabla\tilde{V}_{|v|,t}^\mathrm{l}\|_{L^\infty}
\|x\varPhi_0\|\,dt=O(|v|^{-2\gamma_D+\epsilon})
\label{3.22}
\end{equation}
with $\epsilon>0$. In order to estimate
$$\int_{-\infty}^\infty\|\nabla\tilde{V}_{|v|,t}^\mathrm{l}\|_{L^\infty}
\left\|\left(\int_0^ts(\nabla \tilde{V}_{|v|,s}^\mathrm{l})(ps)\,ds\right)\varPhi_0\right\|\,dt,$$
we consider
\begin{align*}
\tilde{I}_{2,1}={}&\int_{-\infty}^\infty(1+|v|^{\nu_1/(2-\nu_1)}|t|)^{-(\gamma_D+1/2)(2-\nu_1)}\\
&\qquad\times|v|^{-(\gamma_D+1/2)\nu_3}|t|^{2-(\gamma_D+1/2)(2-\nu_3)}\,dt\\
={}&O(|v|^{-(\gamma_D+1/2)\nu_3-\{3-(\gamma_D+1/2)(2-\nu_3)\}\nu_1/(2-\nu_1)}),\\
\tilde{I}_{2,2}={}&\int_{-\infty}^\infty|v|^{\nu_2/(2-\nu_2)-1}(1+|v|^{\nu_2/(2-\nu_2)}|t|)^{-\gamma_D(2-\nu_2)-1}\\
&\qquad\times|v|^{-(\gamma_D+1/2)\nu_4}|t|^{2-(\gamma_D+1/2)(2-\nu_4)}\,dt\\
={}&O(|v|^{\nu_2/(2-\nu_2)-1-(\gamma_D+1/2)\nu_4-\{3-(\gamma_D+1/2)(2-\nu_4)\}\nu_2/(2-\nu_2)})\\
={}&O(|v|^{-1-(\gamma_D+1/2)\nu_4-\{2-(\gamma_D+1/2)(2-\nu_4)\}\nu_2/(2-\nu_2)})
\end{align*}
by \eqref{3.6} and \eqref{3.8},
which can be obtained under the conditions
$-(\gamma_D+1/2)(2-\nu_1)+2-(\gamma_D+1/2)(2-\nu_3)<-1$ and
$-\gamma_D(2-\nu_2)-1+2-(\gamma_D+1/2)(2-\nu_4)<-1$, i.e.
$(\gamma_D+1/2)(\nu_1+\nu_3)<4\gamma_D-1$ and
$\gamma_D\nu_2+(\gamma_D+1/2)\nu_4<4\gamma_D-1$.
Noting
$$2(\gamma_D+1/2)-(4\gamma_D-1)=2-2\gamma_D>1>0$$
by $\gamma_D<1/2$, we see
$$\bigl\{(\gamma_D+1/2)(\nu_1+\nu_3)\bigm|0\le\nu_1,\,\nu_3\le1\bigr\}
\supset[0,4\gamma_D-1].$$
We also note that for $\nu_1$ and $\nu_3$ such that
$(\gamma_D+1/2)(\nu_1+\nu_3)=4\gamma_D-1$,
\begin{align*}
&-(\gamma_D+1/2)\nu_3-\{3-(\gamma_D+1/2)(2-\nu_3)\}\nu_1/(2-\nu_1)\\
={}&-(\gamma_D+1/2)\nu_3-\{2-2\gamma_D+(\gamma_D+1/2)\nu_3\}\\
&\qquad\qquad\qquad\qquad\qquad\times\frac{(4\gamma_D-1)-(\gamma_D+1/2)\nu_3}{2(\gamma_D+1/2)-(4\gamma_D-1)+(\gamma_D+1/2)\nu_3}\\
={}&-(4\gamma_D-1)
\end{align*}
holds. This yields
\begin{equation}
\begin{split}
&\inf_{\nu_1,\,\nu_3}
(-(\gamma_D+1/2)\nu_3-\{3-(\gamma_D+1/2)(2-\nu_3)\}\nu_1/(2-\nu_1))\\
={}&-(4\gamma_D-1).
\end{split}\label{3.23}
\end{equation}
Noting
$$\gamma_D+(\gamma_D+1/2)-(4\gamma_D-1)=3/2-2\gamma_D>1/2>0$$
by $\gamma_D<1/2$, we see
$$\bigl\{\gamma_D\nu_2+(\gamma_D+1/2)\nu_4\bigm|0\le\nu_2,\,\nu_4\le1\bigr\}
\supset[0,4\gamma_D-1].$$
We also note that for $\nu_2$ and $\nu_4$ such that
$\gamma_D\nu_2+(\gamma_D+1/2)\nu_4=4\gamma_D-1$,
\begin{align*}
&-1-(\gamma_D+1/2)\nu_4-\{2-(\gamma_D+1/2)(2-\nu_4)\}\nu_2/(2-\nu_2)\\
={}&-1-(\gamma_D+1/2)\nu_4-\{1-2\gamma_D+(\gamma_D+1/2)\nu_4\}\\
&\qquad\qquad\qquad\qquad\qquad\qquad\times\frac{(4\gamma_D-1)-(\gamma_D+1/2)\nu_4}{2\gamma_D-(4\gamma_D-1)+(\gamma_D+1/2)\nu_4}\\
={}&-1-(4\gamma_D-1)=-4\gamma_D
\end{align*}
holds. This yields
\begin{equation}
\begin{split}
&\inf_{\nu_2,\,\nu_4}
(-1-(\gamma_D+1/2)\nu_4-\{2-(\gamma_D+1/2)(2-\nu_4)\}\nu_2/(2-\nu_2))\\
={}&-4\gamma_D.
\end{split}\label{3.24}
\end{equation}
\eqref{3.23} and \eqref{3.24} imply
\begin{equation}
\int_{-\infty}^\infty\|\nabla\tilde{V}_{|v|,t}^\mathrm{l}\|_{L^\infty}
\left\|\left(\int_0^ts(\nabla \tilde{V}_{|v|,s}^\mathrm{l})(ps)\,ds\right)\varPhi_0\right\|\,dt=O(|v|^{-(4\gamma_D-1)+\epsilon})
\label{3.25}
\end{equation}
with $\epsilon>0$.
In order to estimate
$$\int_{-\infty}^\infty|t|\|\Delta\tilde{V}_{|v|,t}^\mathrm{l}\|_{L^\infty}\|\varPhi_0\|\,dt,$$
we consider
\begin{align*}
\tilde{I}_{3,1}={}&\int_{-\infty}^\infty(1+|v|^{\nu_1/(2-\nu_1)}|t|)^{-(\gamma_D+1)(2-\nu_1)}|t|\,dt=O(|v|^{-2\nu_1/(2-\nu_1)}),\\
\tilde{I}_{3,2}={}&\int_{-\infty}^\infty|v|^{\nu_2/(2-\nu_2)-1}(1+|v|^{\nu_2/(2-\nu_2)}|t|)^{-(\gamma_D+1/2)(2-\nu_2)-1}|t|\,dt\\
={}&O(|v|^{\nu_2/(2-\nu_2)-1-2\nu_2/(2-\nu_2)})=O(|v|^{-1-\nu_2/(2-\nu_2)}),\\
\tilde{I}_{3,3}={}&\int_{-\infty}^\infty|v|^{2\nu_3/(2-\nu_3)-2}(1+|v|^{\nu_3/(2-\nu_3)}|t|)^{-\gamma_D(2-\nu_3)-2}|t|\,dt\\
={}&O(|v|^{2\nu_3/(2-\nu_3)-2-2\nu_3/(2-\nu_3)})=O(|v|^{-2})
\end{align*}
by \eqref{3.7}, which can be obtained by an appropriate change of variables
under the conditions $-(\gamma_D+1)(2-\nu_1)+1<-1$,
$-(\gamma_D+1/2)(2-\nu_2)-1+1<-1$ and
$-\gamma_D(2-\nu_3)-2+1<-1$, i.e.
$\nu_1<2-2/(\gamma_D+1)$,
$\nu_2<2-1/(\gamma_D+1/2)$
and $\nu_3<2$. Since $2/5<2-2/(\gamma_D+1)<2/3$,
$2/3<2-1/(\gamma_D+1/2)<1$,
\begin{align*}
&\inf_{0\le\nu_1<2-2/(\gamma_D+1)}(-2\nu_1/(2-\nu_1))=-2\gamma_D,\\
&\inf_{0\le\nu_2<2-1/(\gamma_D+1/2)}(-1-\nu_2/(2-\nu_2))=-1-2\gamma_D,
\end{align*}
and $-2<-1-2\gamma_D<-2\gamma_D$,
we obtain
\begin{equation}
\int_{-\infty}^\infty|t|\|\Delta\tilde{V}_{|v|,t}^\mathrm{l}\|_{L^\infty}\|\varPhi_0\|\,dt=O(|v|^{-2\gamma_D+\epsilon})
\label{3.26}
\end{equation}
with $\epsilon>0$.
By \eqref{3.20}, \eqref{3.22}, \eqref{3.25} and \eqref{3.26}, we finally obtain
\eqref{3.17} because of $-1<-2\gamma_D<-(4\gamma_D-1)$.
\end{proof}

In the same way as in \cite{AM}, we introduce auxiliary wave operators
$$\Omega_{D,G,v}^{\mathrm{s},\pm}=\slim_{t\to\pm\infty}
e^{itH}U_{D,G,v}^\mathrm{s}(t),$$
where $U_{D,G,v}^\mathrm{s}(t)=U_D(t)M_{G,v}^\mathrm{s}(t)$
and $M_{G,v}^\mathrm{s}(t)=e^{-i\int_0^t V^\mathrm{s}(vs+e_1s^2/2)\,ds}$
as in \S2. Then we see that
$$\Omega_{D,G,v}^{\mathrm{s},\pm}=W_D^\pm I_{G,v}^{\mathrm{s},\pm},\quad
I_{G,v}^{\mathrm{s},\pm}=\slim_{t\to\pm\infty}M_{G,v}^\mathrm{s}(t)$$
exist. Therefore, by Lemmas \ref{lem3.3}, \ref{lem3.4} and \ref{lem3.5},
the following lemma can be
obtained as Lemma \ref{lem2.5}. Thus we omit the proof.

\begin{lem}\label{lem3.6}
Let $v$ and $\varPhi_v$ be as in Theorem \ref{thm3.1}, and $\epsilon>0$.
Then
\begin{equation}
\sup_{t\in\boldsymbol{R}}\|(e^{-itH}\Omega_{D,G,v}^{\mathrm{s},-}-U_{D,G,v}^\mathrm{s}(t))\varPhi_v\|=
O(|v|^{\max\{-1,-2(\gamma_1-1)+\epsilon,-(4\gamma_D-1)+\epsilon\}})
\label{3.27}
\end{equation}
holds as $|v|\to\infty$ for $V^\mathrm{vs}\in\mathscr{V}^\mathrm{vs}$,
$V^\mathrm{s}\in\mathscr{V}^\mathrm{s}(1/2,1)$,
and $V^\mathrm{l}\in\mathscr{V}_D^\mathrm{l}(1/4)$.
\end{lem}

Now we will show Theorem \ref{thm3.1}:

\begin{proof}[Proof of Theorem \ref{thm3.1}]
Since the proof is quite similar to the one of
Theorem \ref{thm2.1}, we give its sketch only.

Suppose that $V^\mathrm{vs}\in\mathscr{V}^\mathrm{vs}$,
$V^\mathrm{s}\in\mathscr{V}^\mathrm{s}(1/2,5/4)$ and
$V^\mathrm{l}\in\mathscr{V}_D^\mathrm{l}(3/8)$.
We first note that $S_D$ is represented as
\begin{align*}
&S_D=(W_D^+)^*W_D^-=I_{G,v}^\mathrm{s}(\Omega_{D,G,v}^{\mathrm{s},+})^*\Omega_{D,G,v}^{\mathrm{s},-},\\
&I_{G,v}^\mathrm{s}=I_{G,v}^{\mathrm{s},+}(I_{G,v}^{\mathrm{s},-})^*=e^{-i\int_{-\infty}^\infty V^\mathrm{s}(vs+e_1s^2/2)\,ds}.
\end{align*}
Noting $[S_D,p_j]=[S_D-I_{G,v}^\mathrm{s},p_j-v_j]$, $(p_j-v_j)\varPhi_v=(p_j\varPhi_0)_v$ and
\begin{align*}
i(S_D-I_{G,v}^\mathrm{s})\varPhi_v={}&I_{G,v}^\mathrm{s}i(\Omega_{D,G,v}^{\mathrm{s},+}-\Omega_{D,G,v}^{\mathrm{s},-})^*
\Omega_{D,G,v}^{\mathrm{s},-}\varPhi_v\\
={}&I_{G,v}^\mathrm{s}\int_{-\infty}^\infty U_{D,G,v}^\mathrm{s}(t)^*V_t^De^{-itH}\Omega_{D,G,v}^{\mathrm{s},-}\varPhi_v\,dt
\end{align*}
with 
$$V_t^D=V^\mathrm{vs}(x)+V^\mathrm{s}(x)-V^\mathrm{s}(vt+e_1t^2/2)
+V^\mathrm{l}(x)-V^\mathrm{l}(pt-e_1t^2/2),$$
we have
$$|v|(i[S_D,p_j]\varPhi_v,\varPsi_v)=I_{G,v}^\mathrm{s}\{I^D(v)+R^D(v)\}$$
with
\begin{align*}
I^D(v)={}&|v|\int_{-\infty}^\infty
[(V_t^DU_{D,G,v}^\mathrm{s}(t)(p_j\varPhi_0)_v,U_{D,G,v}^\mathrm{s}(t)\varPsi_v)\\
{}&\qquad\qquad-(V_t^DU_{D,G,v}^\mathrm{s}(t)\varPhi_v,U_{D,G,v}^\mathrm{s}(t)(p_j\varPsi_0)_v)]\,dt,\\
R^D(v)={}&|v|\int_{-\infty}^\infty
[((e^{-itH}\Omega_{D,G,v}^{\mathrm{s},-}-U_{D,G,v}^\mathrm{s}(t))
(p_j\varPhi_0)_v,V_t^DU_{D,G,v}^\mathrm{s}(t)\varPsi_v)\\
{}&\qquad\qquad-((e^{-itH}\Omega_{D,G,v}^{\mathrm{s},-}-U_{D,G,v}^\mathrm{s}(t))\varPhi_v,V_t^DU_{D,G,v}^\mathrm{s}(t)(p_j\varPsi_0)_v)]\,dt.
\end{align*}
By Lemmas \ref{lem3.3}, \ref{lem3.4}, \ref{lem3.5} and \ref{lem3.6},
we have
\begin{equation}
\begin{split}
|R^D(v)|={}&O(|v|^{1+2\max\{-1,-2(\gamma_1-1)+\epsilon,-(4\gamma_D-1)+\epsilon\}})\\
={}&O(|v|^{\max\{-1,5-4\gamma_1+2\epsilon,3-8\gamma_D+2\epsilon\}}).
\end{split}\label{3.28}
\end{equation}\
In the same way as in the
proof of Theorem \ref{thm2.1}, we need the conditions
$5-4\gamma_1<0$ and $3-8\gamma_D<0$,
which are equivalent to
$\gamma_1>5/4$ and $\gamma_D>3/8$,
in order to get $\lim_{|v|\to\infty}R^D(v)=0$.

The rest of the proof is the same as in \cite{We1} and \cite{AM}.
So we omit it.
\end{proof}

By virtue of Theorem \ref{thm3.1},
Theorem \ref{thm1.2} can be shown in the same way as in the
proof of Theorem \ref{thm1.1}. Thus we omit its proof.

\section{The case where $V^\mathrm{s}\in\tilde{\mathscr{V}}^\mathrm{s}(1/2,1,5/4)$}

Throughout this section, we suppose
$V^\mathrm{s}\in\tilde{\mathscr{V}}^\mathrm{s}(1/2,1,1)$.
Then the following reconstruction
formulas, which are Theorems \ref{thm2.1} and \ref{thm3.1}
with replacing $V^\mathrm{s}\in\mathscr{V}^\mathrm{s}(1/2,5/4)$
by $V^\mathrm{s}\in\tilde{\mathscr{V}}^\mathrm{s}(1/2,1,5/4)$,
can be obtained:

\begin{thm}\label{thm4.1} Let the notation in this theorem be the same
as in Theorem \ref{thm2.1}.
Let $V^\mathrm{vs}\in\mathscr{V}^\mathrm{vs}$,
$V^\mathrm{s}\in\tilde{\mathscr{V}}^\mathrm{s}(1/2,1,5/4)$.
Then \eqref{2.1} holds for $1\le j\le n$.
\end{thm}

\begin{thm}\label{thm4.2}
Let the notation in this theorem be the same
as in Theorem \ref{thm3.1}.
Let $V^\mathrm{vs}\in\mathscr{V}^\mathrm{vs}$,
$V^\mathrm{s}\in\tilde{\mathscr{V}}^\mathrm{s}(1/2,1,5/4)$,
$V^\mathrm{l}\in\mathscr{V}_D^\mathrm{l}(3/8)$.
Then \eqref{3.2} holds for $1\le j\le n$.
\end{thm}

In order to prove Theorems \ref{thm4.1} and \ref{thm4.2}, we would like to
improve a series of lemmas in \S2 and \S3,
for $V^\mathrm{s}\in\tilde{\mathscr{V}}^\mathrm{s}(1/2,1,1)$. To this end, we will introduce
\begin{equation}
U_D^\mathrm{s}(t)=e^{-itH_0}M_D^\mathrm{s}(t),\quad
M_D^\mathrm{s}(t)=e^{-i\int_0^tV^\mathrm{s}(ps+e_1s^2/2)\,ds}.
\end{equation}
In \cite{N1}, instead of $M_D^\mathrm{s}(t)$,
$e^{-i\int_0^tV^\mathrm{s}(p_\perp s+e_1s^2/2)\,ds}$ was used,
as mentioned in \S1.
$M_D^\mathrm{s}(t)$ seems more appropriate
for the problem considered in this paper
than $e^{-i\int_0^tV^\mathrm{s}(p_\perp s+e_1s^2/2)\,ds}$.
We first give the following lemma:

\begin{lem}\label{lem4.3}
Let $v$ and $\varPhi_v$ be as in Theorem \ref{thm4.1},
and $V^\mathrm{s}\in\tilde{\mathscr{V}}^\mathrm{s}(1/2,1,1)$.
If $\gamma_2>3/2$,
then there exists a positive constant $C$ such that
\begin{equation}
\|\langle x\rangle^2M_D^\mathrm{s}(t)\varPhi_v\|
=\|\langle x\rangle^2M_{D,v}^\mathrm{s}(t)\varPhi_0\|
\le C
\label{4.2}
\end{equation}
holds as $|v|\to\infty$, where
$M_{D,v}^\mathrm{s}(t)=e^{-iv\cdot x}M_D^\mathrm{s}(t)e^{iv\cdot x}=e^{-i\int_0^tV^\mathrm{s}(ps+vs+e_1s^2/2)\,ds}$.
On the other hand, if $\gamma_2\le3/2$, then, for $0\le\nu_1\le1$,
there exists a positive constant $C$ such that
\begin{equation}
\|\langle x\rangle^2M_D^\mathrm{s}(t)\varPhi_v\|
=\|\langle x\rangle^2M_{D,v}^\mathrm{s}(t)\varPhi_0\|
\le C(1+|v|^{-\gamma_2\nu_1}|t|^{3-\gamma_2(2-\nu_1)})
\end{equation}
holds as $|v|\to\infty$, where only when $\gamma_2=3/2$,
we assume $\nu_1\not=0$ additionally.
\end{lem}

\begin{proof}
Since the proof is quite similar to the one of Lemma \ref{lem3.2},
we will sketch it.

We first note that
since $\mathrm{supp}\,\hat{\varPhi}_0\subset
\bigl\{\xi\in\boldsymbol{R}^n\bigm||\xi|<\eta\bigr\}$,
\begin{equation}
M_{D,v}^\mathrm{s}(t)\varPhi_0
=e^{-i\int_0^t\tilde{V}_{|v|,s}^\mathrm{s}(ps)\,ds}\varPhi_0
\label{4.4}
\end{equation}
holds for $|v|\ge\eta/(3\lambda_1)$.
Hence, as in the proof of Lemma \ref{lem3.2},
\begin{align*}
&\|\langle x\rangle^2M_{D,v}^\mathrm{s}(t)\varPhi_0\|\\
\le{}&\|\langle x\rangle^2\varPhi_0\|+
\left\|\left(\int_0^ts^2(\Delta \tilde{V}_{|v|,s}^\mathrm{s})(ps)\,ds\right)\varPhi_0\right\|\\
&+2\left\|\left(\int_0^ts(\nabla \tilde{V}_{|v|,s}^\mathrm{s})(ps)\,ds\right)\cdot x\varPhi_0\right\|+
\left\|\left(\int_0^ts(\nabla \tilde{V}_{|v|,s}^\mathrm{s})(ps)\,ds\right)^2\varPhi_0\right\|
\end{align*}
can be obtained.
Now we will estimate $\|\nabla\tilde{V}_{|v|,t}^\mathrm{s}\|_{L^\infty}$
and $\|\Delta\tilde{V}_{|v|,t}^\mathrm{s}\|_{L^\infty}$.
In the same way as in the proof of Lemma \ref{lem3.2},
for $0\le\nu_1,\,\nu_2,\,\nu_3,\,\nu_4,\,\nu_5\le1$ and $|v|\ge1$,
\begin{equation}
\begin{split}
\|\nabla\tilde{V}_{|v|,t}^\mathrm{s}\|_{L^\infty}
\le{}&
C_1'(1+|v|^{\nu_1/(2-\nu_1)}|t|)^{-\gamma_1(2-\nu_1)}\\
&+C_2'|v|^{\nu_2/(2-\nu_2)-1}(1+|v|^{\nu_2/(2-\nu_2)}|t|)^{-\gamma_0(2-\nu_2)-1}
\end{split}\label{4.5}
\end{equation}
and
\begin{equation}
\begin{split}
\|\Delta\tilde{V}_{|v|,t}^\mathrm{s}\|_{L^\infty}
\le{}&
C_3'(1+|v|^{\nu_3/(2-\nu_3)}|t|)^{-\gamma_2(2-\nu_3)}\\
&+C_4'|v|^{\nu_4/(2-\nu_4)-1}(1+|v|^{\nu_4/(2-\nu_4)}|t|)^{-\gamma_1(2-\nu_4)-1}\\
&\quad+C_5'|v|^{2\nu_5/(2-\nu_5)-2}(1+|v|^{\nu_5/(2-\nu_5)}|t|)^{-\gamma_0(2-\nu_5)-2}
\end{split}\label{4.6}
\end{equation}
can be obtained (cf. \eqref{3.6} and \eqref{3.7}).
Since $-2\gamma_1+1<-1$ and $-2\gamma_0-1+1<-1$
by assumption, the estimate
\begin{equation}
\left\|\int_0^ts(\nabla \tilde{V}_{|v|,s}^\mathrm{s})(ps)\,ds\right\|_{\mathscr{B}(L^2)}\le C
\label{4.7}
\end{equation}
can be obtained immediately by \eqref{4.4} with $\nu_1=\nu_2=0$.
Since $-2\gamma_1-1+2<-1$ and $-2\gamma_0-2+2<-1$
by assumption, we see that
$|t|^2\times$(the second and third terms of the right-hand side
of \eqref{4.6} with $\nu_4=\nu_5=0$) are integrable
in $\boldsymbol{R}$. Hence,
we have only to watch
$$\tilde{I}=C_3'\int_0^{|t|}s^2(1+|v|^{\nu_3/(2-\nu_3)}s)^{-\gamma_2(2-\nu_3)}\,ds.$$
If $\gamma_2>3/2$, then $-2\gamma_2+2<-1$ holds, which implies that
there exists $C>0$ independent of $t$ such that $\tilde{I}\le C$ holds,
by taking $\nu_3=0$. On the other hand, if $\gamma_2\le3/2$,
then
$$\tilde{I}\le C_3''|v|^{-\gamma_2\nu_3}|t|^{3-\gamma_2(2-\nu_3)}$$
can be obtained easily, where only when $\gamma_2=3/2$,
we assume $\nu_3\not=0$ additionally.

Based on the above observations, the lemma can be proved.
\end{proof}

By virtue of Lemma \ref{lem4.3},
the following lemma can be obtained in the same way as in the proof of Lemma \ref{lem3.3}.

\begin{lem}\label{lem4.4} Let $v$ and $\varPhi_v$ be as
in Theorem \ref{thm4.1}, and $V^\mathrm{s}\in\tilde{\mathscr{V}}^\mathrm{s}(1/2,1,1)$. Then
\begin{equation}
\int_{-\infty}^\infty
\|V^\mathrm{vs}(x)M_D^\mathrm{s}(t)\varPhi_v\|\,dt
=O(|v|^{-1})
\end{equation}
holds as $|v|\to\infty$ for $V^\mathrm{vs}\in\mathscr{V}^\mathrm{vs}$.
\end{lem}

\begin{proof}
We have only to 
mention some changes compared to the proof of Lemma \ref{lem3.3}:
We consider the case where $\gamma_2<3/2$ only.
The estimate $I_2$ in the proof of Lemma \ref{lem3.3} has to be
replaced by
$$I_2\le C(1+\lambda_1|v||t|)^{-2}(1+|v|^{-\gamma_2\nu_1}|t|^{3-\gamma_2(2-\nu_1)})$$
for $0\le\nu_1\le1$.
Therefore
$$\int_{-\infty}^\infty I_2\,dt=O(|v|^{-1})+
O(|v|^{-\gamma_2\nu_1-\{4-\gamma_2(2-\nu_1)\}})$$
can be obtained by an appropriate change of variables
under the condition
$-2+3-\gamma_2(2-\nu_1)<-1$, i.e.
$\nu_1<2-2/\gamma_2$.
Since $0<2-2/\gamma_2<2/3$ and
$$\inf_{0\le\nu_1<2-2/\gamma_2}(-\gamma_2\nu_1-\{4-\gamma_2(2-\nu_1)\})=-2\gamma_2,$$
we obtain
$$\int_{-\infty}^\infty I_2\,dt=O(|v|^{-1}).$$

Based on the above observations, the
lemma can be proved.
\end{proof}

The following lemma can be obtained as in the proof of Lemma \ref{lem3.5}:

\begin{lem}\label{lem4.5}
Let $v$ and $\varPhi_v$ be as in Theorem \ref{thm4.1}, $\epsilon>0$,
and $V^\mathrm{s}\in\tilde{\mathscr{V}}^\mathrm{s}(1/2,1,1)$.
Then
\begin{equation}
\int_{-\infty}^\infty
\|\{V^\mathrm{s}(x)-V^\mathrm{s}(pt-e_1t^2/2)\}U_D^\mathrm{s}(t)\varPhi_v\|\,dt
=O(|v|^{\max\{-1,-2(\gamma_2-1)+\epsilon\}})
\label{4.9}
\end{equation}
holds as $|v|\to\infty$.
\end{lem}

\begin{proof}
We have only to 
mention some changes compared to the proof of Lemma \ref{lem3.5}:
For the sake of brevity, we put $I=\|\{V^\mathrm{s}(x)-V^\mathrm{s}(pt-e_1t^2/2)\}U_D^\mathrm{s}(t)\varPhi_v\|$.
$I$ can be estimated as
\begin{equation}
\begin{split}
I={}&\|\{V^\mathrm{s}(x+vt+e_1t^2/2)-V^\mathrm{s}(pt+vt+e_1t^2/2)\}e^{-itK_0}M_{D,v}^\mathrm{s}(t)\varPhi_0\|\\
\le{}&\|\bar{V}_{|v|,t}^\mathrm{s}(x)e^{-itK_0}
f(p)M_{D,v}^\mathrm{s}(t)\varPhi_0\|\\
&+
\|\{\tilde{V}_{|v|,t}^\mathrm{s}(x)-\tilde{V}_{|v|,t}^\mathrm{s}(pt)\}e^{-itK_0}
M_{D,v}^\mathrm{s}(t)\varPhi_0\|
\end{split}\label{4.10}
\end{equation}
for $|v|\ge\eta/(3\lambda_1)$.
As for the first term in the inequality \eqref{4.10}, the estimate
\begin{equation}
\int_{-\infty}^\infty\|\bar{V}_{|v|,t}^\mathrm{s}(x)e^{-itK_0}
f(p)M_{D,v}^\mathrm{s}(t)\varPhi_0\|\,dt=O(|v|^{-1})
\label{4.11}
\end{equation}
can be obtained as in the proof of Lemma \ref{lem3.5}.
On the other hand, as for the second term of the inequality \eqref{4.10},
we see that
\begin{align*}
&\|\{\tilde{V}_{|v|,t}^\mathrm{s}(x)-\tilde{V}_{|v|,t}^\mathrm{s}(pt)\}e^{-itK_0}
M_{D,v}^\mathrm{s}(t)\varPhi_0\|\\
\le{}&\|\nabla\tilde{V}_{|v|,t}^\mathrm{s}\|_{L^\infty}\left(\|x\varPhi_0\|+\left\|\left(\int_0^ts(\nabla \tilde{V}_{|v|,s}^\mathrm{s})(ps)\,ds\right)\varPhi_0\right\|\right)\\
&+\frac{1}{2}|t|\|\Delta\tilde{V}_{|v|,t}^\mathrm{s}\|_{L^\infty}\|\varPhi_0\|\\\le{}&C_1\|\nabla\tilde{V}_{|v|,t}^\mathrm{s}\|_{L^\infty}
+C_2|t|\|\Delta\tilde{V}_{|v|,t}^\mathrm{s}\|_{L^\infty}
\end{align*}
holds, by virtue of the Baker-Campbell-Hausdorff formula.
Here we used \eqref{4.7}.
The estimate
\begin{equation}
\int_{-\infty}^\infty C_1\|\nabla\tilde{V}_{|v|,t}^\mathrm{s}\|_{L^\infty}\,dt=O(|v|^{-1})
\label{4.12}
\end{equation}
can be obtained in the same way as in the proof of Lemma \ref{lem2.4}
(cf. \eqref{2.15}).
In order to estimate
$$\int_{-\infty}^\infty C_2|t|\|\Delta\tilde{V}_{|v|,t}^\mathrm{s}\|_{L^\infty}
\,dt,$$
we consider
\begin{align*}
\tilde{I}_{2,1}={}&\int_{-\infty}^\infty(1+|v|^{\nu_1/(2-\nu_1)}|t|)^{-\gamma_2(2-\nu_1)}|t|\,dt=O(|v|^{-2\nu_1/(2-\nu_1)}),\\
\tilde{I}_{2,2}={}&\int_{-\infty}^\infty|v|^{\nu_2/(2-\nu_2)-1}(1+|v|^{\nu_2/(2-\nu_2)}|t|)^{-\gamma_1(2-\nu_2)-1}|t|\,dt\\
={}&O(|v|^{\nu_2/(2-\nu_2)-1-2\nu_2/(2-\nu_2)})=O(|v|^{-1-\nu_2/(2-\nu_2)}),\\
\tilde{I}_{2,3}={}&\int_{-\infty}^\infty|v|^{2\nu_3/(2-\nu_3)-2}(1+|v|^{\nu_3/(2-\nu_3)}|t|)^{-\gamma_0(2-\nu_3)-2}|t|\,dt\\
={}&O(|v|^{2\nu_3/(2-\nu_3)-2-2\nu_3/(2-\nu_3)})=O(|v|^{-2})
\end{align*}
by \eqref{4.6}, which can be obtained by an appropriate
change of variables under the conditions
$-\gamma_2(2-\nu_1)+1<-1$,
$-\gamma_1(2-\nu_2)-1+1<-1$ and
$-\gamma_0(2-\nu_3)-2+1<-1$, i.e.
$\nu_1<2-2/\gamma_2$, $\nu_2<2-1/\gamma_1$ and
$\nu_3<2$. Since $0<2-2/\gamma_2<2/3$,
$2-1/\gamma_1>1$,
$$\inf_{0\le\nu_1<2-2/\gamma_2}(-2\nu_1/(2-\nu_1))=-2(\gamma_2-1),$$
and $-2(\gamma_2-1)\ge-2\gamma_1>-2$, we have
\begin{equation}
\int_{-\infty}^\infty|t|\|\Delta\tilde{V}_{|v|,t}^\mathrm{s}\|_{L^\infty}\|\varPhi_0\|\,dt=O(|v|^{-2(\gamma_2-1)+\epsilon})
\label{4.13}
\end{equation}
with $\epsilon>0$.
By \eqref{4.11}, \eqref{4.12} and \eqref{4.13}, \eqref{4.9} can be obtained.
\end{proof}

By modifying the argument in \cite{N1}, we introduce auxiliary wave operators
$$\Omega_D^{\mathrm{s},\pm}=\slim_{t\to\pm\infty}
e^{itH}U_D^\mathrm{s}(t)$$
with $H=H_0+V^\mathrm{vs}+V^\mathrm{s}$.
Then we see that
$$\Omega_D^{\mathrm{s},\pm}=W^\pm I_D^{\mathrm{s},\pm},\quad
I_D^{\mathrm{s},\pm}=\slim_{t\to\pm\infty}M_D^\mathrm{s}(t)=
e^{-i\int_0^{\pm\infty} V^\mathrm{s}(ps+e_1s^2/2)\,ds}$$
exist. Therefore, by Lemmas \ref{lem4.4} and \ref{lem4.5},
the following lemma can be
obtained as an improvement of Lemma \ref{lem2.5}. Thus we omit the proof.

\begin{lem}\label{lem4.6}
Let $v$ and $\varPhi_v$ be as in Theorem \ref{thm4.1} and $\epsilon>0$.
Then
\begin{equation}
\sup_{t\in\boldsymbol{R}}\|(e^{-itH}\Omega_D^{\mathrm{s},-}-U_D^\mathrm{s}(t))\varPhi_v\|=
O(|v|^{\max\{-1,-2(\gamma_2-1)+\epsilon\}})
\end{equation}
holds as $|v|\to\infty$ for $V^\mathrm{vs}\in\mathscr{V}^\mathrm{vs}$
and
$V^\mathrm{s}\in\tilde{\mathscr{V}}^\mathrm{s}(1/2,1,1)$.
\end{lem}

Now we will show Theorem \ref{thm4.1}:

\begin{proof}[Proof of Theorem \ref{thm4.1}]
Since the proof is quite similar to the one of
Theorem \ref{thm2.1}, we give its sketch only.

Suppose that $V^\mathrm{vs}\in\mathscr{V}^\mathrm{vs}$
and $V^\mathrm{s}\in\tilde{\mathscr{V}}^\mathrm{s}(1/2,1,5/4)$.
We first note that $S$ is represented as
$$S=(W^+)^*W^-=I_D^{\mathrm{s},+}(\Omega_D^{\mathrm{s},+})^*\Omega_D^{\mathrm{s},-}(I_D^{\mathrm{s},-})^*$$
Unlike $I_{G,v}^{\mathrm{s},\pm}$, $I_D^{\mathrm{s},\pm}$ do not
commute with $\Omega_D^{\mathrm{s},\pm}$.
Putting
$$I_D^\mathrm{s}=I_D^{\mathrm{s},+}(I_D^{\mathrm{s},-})^*=
e^{-i\int_{-\infty}^\infty V^\mathrm{s}(ps+e_1s^2/2)\,ds}$$
and noting $[S,p_j]=[S-I_D^\mathrm{s},p_j-v_j]$, $(p_j-v_j)\varPhi_v=(p_j\varPhi_0)_v$ and
\begin{align*}
i(S-I_D^\mathrm{s})\varPhi_v={}&I_D^{\mathrm{s},+}i(\Omega_D^{\mathrm{s},+}-\Omega_D^{\mathrm{s},-})^*\Omega_D^{\mathrm{s},-}(I_D^{\mathrm{s},-})^*\varPhi_v\\
={}&I_D^{\mathrm{s},+}\int_{-\infty}^\infty U_D^\mathrm{s}(t)^*V_{t,D}e^{-itH}\Omega_D^{\mathrm{s},-}((I_{D,v}^{\mathrm{s},-})^*\varPhi_0)_v\,dt
\end{align*}
with
\begin{align*}
&V_{t,D}=V^\mathrm{vs}(x)+V^\mathrm{s}(x)-V^\mathrm{s}(pt-e_1t^2/2),\\
&(I_{D,v}^{\mathrm{s},\pm})^*=e^{-iv\cdot x}(I_D^{\mathrm{s},\pm})^*e^{iv\cdot x}=
e^{i\int_0^{\pm\infty}V^\mathrm{s}(ps+vs+e_1s^2/2)\,ds},
\end{align*}
we have
$$|v|(i[S,p_j]\varPhi_v,\varPsi_v)=I_D(v)+R_D(v)$$
with
\begin{align*}
&I_D(v)\\
={}&|v|\int_{-\infty}^\infty
[(V_{t,D}U_D^\mathrm{s}(t)(p_j(I_{D,v}^{\mathrm{s},-})^*\varPhi_0)_v,U_D^\mathrm{s}(t)((I_{D,v}^{\mathrm{s},+})^*\varPsi_0)_v)\\
{}&\qquad-(V_{t,D}U_D^\mathrm{s}(t)((I_{D,v}^{\mathrm{s},-})^*\varPhi_0)_v,U_D^\mathrm{s}(t)(p_j(I_{D,v}^{\mathrm{s},+})^*\varPsi_0)_v]\,dt,\\
&R_D(v)\\
={}&|v|\int_{-\infty}^\infty
[(e^{-itH}\Omega_D^{\mathrm{s},-}-U_D^\mathrm{s}(t))(p_j(I_{D,v}^{\mathrm{s},-})^*\varPhi_0)_v,V_{t,D}U_D^\mathrm{s}(t)((I_{D,v}^{\mathrm{s},+})^*\varPsi_0)_v)\\
{}&\qquad-(e^{-itH}\Omega_D^{\mathrm{s},-}-U_D^\mathrm{s}(t))((I_{D,v}^{\mathrm{s},-})^*\varPhi_0)_v,V_{t,D}U_D^\mathrm{s}(t)(p_j(I_{D,v}^{\mathrm{s},+})^*\varPsi_0)_v]\,dt.
\end{align*}
Here we used that $p_j$ commutes with $(I_{D,v}^{\mathrm{s},\pm})^*$.
By Lemmas \ref{lem4.4}, \ref{lem4.5} and \ref{lem4.6},
we have 
\begin{equation}
|R_D(v)|=O(|v|^{1+2\max\{-1,-2(\gamma_2-1)+\epsilon\}})
=O(|v|^{\max\{-1,5-4\gamma_2+2\epsilon\}}).
\label{4.16}
\end{equation}
Here we used $\mathrm{supp}\,\mathscr{F}[(I_{D,v}^{\mathrm{s},-})^*\varPhi_0]
=\mathrm{supp}\,\mathscr{F}[\varPhi_0]$ and
$\mathrm{supp}\,\mathscr{F}[(I_{D,v}^{\mathrm{s},+})^*\varPsi_0]
=\mathrm{supp}$
$\mathscr{F}[\varPsi_0]$.
Then we need the condition $5-4\gamma_2<0$
in order to get $\lim_{|v|\to\infty}R_D(v)=0$.
Using the Avron-Herbst formula \eqref{2.8} and \eqref{2.9},
$I_D(v)$ is rewritten as
\begin{align*}
I_D(v)
={}&|v|\int_{-\infty}^\infty
[(\hat{V}_{t,D}e^{-itK_0}M_{D,v}^\mathrm{s}(t)p_j(I_{D,v}^{\mathrm{s},-})^*\varPhi_0,e^{-itK_0}M_{D,v}^\mathrm{s}(t)(I_{D,v}^{\mathrm{s},+})^*\varPsi_0)\\
{}&\qquad-(\hat{V}_{t,D}e^{-itK_0}M_{D,v}^\mathrm{s}(t)(I_{D,v}^{\mathrm{s},-})^*\varPhi_0,e^{-itK_0}M_{D,v}^\mathrm{s}(t)p_j(I_{D,v}^{\mathrm{s},+})^*\varPsi_0]\,dt
\end{align*}
with
$$\hat{V}_{t,D}
=V^\mathrm{vs}(x+vt+e_1t^2/2)+V^\mathrm{s}(x+vt+e_1t^2/2)-V^\mathrm{s}(pt+vt+e_1t^2/2).$$
Since
$$[V^\mathrm{s}(x+vt+e_1t^2/2)-V^\mathrm{s}(pt+vt+e_1t^2/2),p_j]\\
=i(\partial_jV^\mathrm{s})(x+vt+e_1t^2/2)$$
and $-\gamma_1<-1$,
$I_D(v)$ is rewritten as
\begin{align*}
I_D(v)
={}&|v|\int_{-\infty}^\infty
[(V^\mathrm{vs}(x+vt+e_1t^2/2)e^{-itK_0}M_{D,v}^\mathrm{s}(t)p_j(I_{D,v}^{\mathrm{s},-})^*\varPhi_0,\\
&\qquad\qquad\qquad\qquad\qquad\qquad\qquad e^{-itK_0}M_{D,v}^\mathrm{s}(t)(I_{D,v}^{\mathrm{s},+})^*\varPsi_0)\\
{}&\qquad-(V^\mathrm{vs}(x+vt+e_1t^2/2)e^{-itK_0}M_{D,v}^\mathrm{s}(t)(I_{D,v}^{\mathrm{s},-})^*\varPhi_0,\\
&\qquad\qquad\qquad\qquad\qquad\qquad\qquad e^{-itK_0}M_{D,v}^\mathrm{s}(t)p_j(I_{D,v}^{\mathrm{s},+})^*\varPsi_0)]\,dt\\
&+|v|\int_{-\infty}^\infty
(i(\partial_jV^\mathrm{s})(x+vt+e_1t^2/2)e^{-itK_0}M_{D,v}^\mathrm{s}(t)(I_{D,v}^{\mathrm{s},-})^*\varPhi_0,\\
&\qquad\qquad\qquad\qquad\qquad\qquad\qquad e^{-itK_0}M_{D,v}^\mathrm{s}(t)(I_{D,v}^{\mathrm{s},+})^*\varPsi_0)\,dt\\
={}&\int_{-\infty}^\infty l_{D,v}(\tau)\,d\tau
\end{align*}
with
\begin{align*}
&l_{D,v}(\tau)\\
={}&
(V^\mathrm{vs}(x+\hat{v}\tau+e_1(\tau/|v|)^2/2)e^{-i(\tau/|v|)K_0}M_{D,v}^\mathrm{s}(\tau/|v|)p_j(I_{D,v}^{\mathrm{s},-})^*\varPhi_0,\\
&\qquad\qquad\qquad\qquad\qquad\qquad\qquad e^{-i(\tau/|v|)K_0}M_{D,v}^\mathrm{s}(\tau/|v|)(I_{D,v}^{\mathrm{s},+})^*\varPsi_0)\\
{}&-(V^\mathrm{vs}(x+\hat{v}\tau+e_1(\tau/|v|)^2/2)e^{-i(\tau/|v|)K_0}M_{D,v}^\mathrm{s}(\tau/|v|)(I_{D,v}^{\mathrm{s},-})^*\varPhi_0,\\
&\qquad\qquad\qquad\qquad\qquad\qquad\qquad\quad e^{-i(\tau/|v|)K_0}M_{D,v}^\mathrm{s}(\tau/|v|)p_j(I_{D,v}^{\mathrm{s},+})^*\varPsi_0)\\
{}&\quad+i((\partial_j V^\mathrm{s})(x+\hat{v}\tau+e_1(\tau/|v|)^2/2)e^{-i(\tau/|v|)K_0}M_{D,v}^\mathrm{s}(\tau/|v|)(I_{D,v}^{\mathrm{s},-})^*\varPhi_0,\\
&\qquad\qquad\qquad\qquad\qquad\qquad\qquad\qquad e^{-i(\tau/|v|)K_0}M_{D,v}^\mathrm{s}(\tau/|v|)(I_{D,v}^{\mathrm{s},+})^*\varPsi_0).
\end{align*}
Here we note that
\begin{equation}
\slim_{|v|\to\infty}(I_{D,v}^{\mathrm{s},\pm})^*=\mathrm{Id}.
\label{4.17}
\end{equation}
In fact, for $|\xi|<\eta$ and $|v|\ge\eta/(4\lambda_1)$,
$$|\xi s+vs+e_1s^2/2|\ge\max\{c_1|v||s|,c_2|s|^2\}$$
holds by \eqref{2.12}. This yields that $V^\mathrm{s}(\xi s+vs+e_1s^2/2)$ is integrable
on $\boldsymbol{R}$ because of $\gamma_0>1/2$, and that
$\lim_{|v|\to\infty}V^\mathrm{s}(\xi s+vs+e_1s^2/2)=0$ for any
$s\not=0$. Hence, the Lebesgue dominated convergence theorem
yields
$$\lim_{|v|\to\infty}\int_0^{\pm\infty}V^\mathrm{s}(\xi s+vs+e_1s^2/2)\,ds=0,$$
which implies \eqref{4.17}. We also note
$$\slim_{|v|\to\infty}e^{-i(\tau/|v|)K_0}M_{D,v}^\mathrm{s}(\tau/|v|)=\mathrm{Id},$$
which can be shown easily. Since in the proof of Lemma \ref{lem2.3},
$|l_{D,v}(\tau)|$ can be estimated as
\begin{align*}
|l_{D,v}(\tau)|\le{}&C\{\|V^\mathrm{vs}(x)(1+K_0)^{-1}
F(|x|\ge\lambda_1|\tau|)\|_{\mathscr{B}(L^2)}\\
&\qquad\qquad\qquad+(1+|\tau|)^{-2}+(1+|\tau|)^{-\gamma_1}\},
\end{align*}
whose right-hand side is $|v|$-independent and integrable on $\boldsymbol{R}$.
Therefore we obtain
\begin{align*}
\lim_{|v|\to\infty}I_D(v)=
\int_{-\infty}^\infty[{}&
(V^\mathrm{vs}(x+\hat{v}\tau)p_j\varPhi_0,\varPsi_0)-
(V^\mathrm{vs}(x+\hat{v}\tau)\varPhi_0,p_j\varPsi_0)\\
{}&\quad+i((\partial_j V^\mathrm{s})(x+\hat{v}\tau)\varPhi_0,\varPsi_0)]\,d\tau
\end{align*}
by the Lebesgue dominated convergence theorem. This yields the theorem.
\end{proof}

As for Theorem \ref{thm4.2}, we need the following lemmas,
which are the versions of Lemmas \ref{lem3.2}, \ref{lem3.3},
\ref{lem3.4}, \ref{lem3.5} and \ref{lem3.6}
in the case where $V^\mathrm{s}\in\tilde{\mathscr{V}}^\mathrm{s}(1/2,1,1)$.
Since their proofs are yielded by the above observations,
we omit them: We will introduce
\begin{equation}
U_{D,D}^\mathrm{s}(t)=e^{-itH_0}M_{D,D}^\mathrm{s}(t),\quad
M_{D,D}^\mathrm{s}(t)=M_D(t)M_D^\mathrm{s}(t).
\end{equation}
Here we note that $M_D(t)$ does commute with $M_D^\mathrm{s}(t)$.

\begin{lem}\label{lem4.7}
Let $v$ and $\varPhi_v$ be as in Theorem \ref{thm4.2},
$V^\mathrm{s}\in\tilde{\mathscr{V}}^\mathrm{s}(1/2,1,1)$,
and $V^\mathrm{l}\in\mathscr{V}_D^\mathrm{l}(1/4)$ with $\gamma_D<1/2$.
If $\gamma_2>3/2$,
then, for $0\le\nu_1,\,\nu_2,\,\nu_3\le1$,
there exists a positive constant $C$ such that
\begin{equation}
\begin{split}
&\|\langle x\rangle^2M_{D,D}^\mathrm{s}(t)\varPhi_v\|
=\|\langle x\rangle^2M_{D,D,v}^\mathrm{s}(t)\varPhi_0\|\\
\le{}&
C(1+|v|^{-(2\gamma_D+1)\nu_1}|t|^{4-(2\gamma_D+1)(2-\nu_1)}
+|v|^{-(\gamma_D+1)\nu_2}|t|^{3-(\gamma_D+1)(2-\nu_2)}\\
&\qquad\qquad
+|v|^{-(\gamma_D+1/2)\nu_3}|t|^{2-(\gamma_D+1/2)(2-\nu_3)})
\end{split}
\end{equation}
holds as $|v|\to\infty$, where
$M_{D,D,v}^\mathrm{s}(t)=e^{-iv\cdot x}M_{D,D}^\mathrm{s}(t)e^{iv\cdot x}$.
On the other hand, if $\gamma_2\le3/2$, then, for $0\le\nu_1,\,\nu_2,\,\nu_3,\,\nu_4\le1$,
there exists a positive constant $C$ such that
\begin{equation}
\begin{split}
&\|\langle x\rangle^2M_{D,D}^\mathrm{s}(t)\varPhi_v\|
=\|\langle x\rangle^2M_{D,D,v}^\mathrm{s}(t)\varPhi_0\|\\
\le{}&
C(1+|v|^{-(2\gamma_D+1)\nu_1}|t|^{4-(2\gamma_D+1)(2-\nu_1)}
+|v|^{-(\gamma_D+1)\nu_2}|t|^{3-(\gamma_D+1)(2-\nu_2)}\\
&\qquad\qquad
+|v|^{-(\gamma_D+1/2)\nu_3}|t|^{2-(\gamma_D+1/2)(2-\nu_3)}
+|v|^{-\gamma_2\nu_4}|t|^{3-\gamma_2(2-\nu_4)})
\end{split}
\end{equation}
holds as $|v|\to\infty$, where only when $\gamma_2=3/2$,
we assume $\nu_1\not=0$ additionally.
\end{lem}

\begin{lem}\label{lem4.8}
Let $v$ and $\varPhi_v$ be as in Theorem \ref{thm4.2},
$V^\mathrm{s}\in\tilde{\mathscr{V}}^\mathrm{s}(1/2,1,1)$,
and $V^\mathrm{l}\in\mathscr{V}_D^\mathrm{l}(1/4)$. Then
\begin{equation}
\int_{-\infty}^\infty
\|V^\mathrm{vs}(x)U_{D,D}^\mathrm{s}(t)\varPhi_v\|\,dt
=O(|v|^{-1})
\end{equation}
holds as $|v|\to\infty$ for $V^\mathrm{vs}\in\mathscr{V}^\mathrm{vs}$.
\end{lem}

\begin{lem}\label{lem4.9}
Let $v$ and $\varPhi_v$ be as in Theorem \ref{thm4.2},
$\epsilon>0$,
$V^\mathrm{s}\in\tilde{\mathscr{V}}^\mathrm{s}(1/2,1,1)$,
and $V^\mathrm{l}\in\mathscr{V}_D^\mathrm{l}(1/4)$.
Then
\begin{equation}
\begin{split}
&\int_{-\infty}^\infty
\|\{V^\mathrm{s}(x)-V^\mathrm{s}(pt-e_1t^2/2)\}U_{D,D}^\mathrm{s}(t)\varPhi_v\|\,dt\\
={}&O(|v|^{\max\{-1,-2(\gamma_2-1)+\epsilon,-2\{(\gamma_1-1)+\gamma_D\}+\epsilon\}})
\end{split}
\end{equation}
holds as $|v|\to\infty$.
\end{lem}

\begin{lem}\label{lem4.10}
Let $v$ and $\varPhi_v$ be as in Theorem \ref{thm4.2},
$\epsilon>0$,
$V^\mathrm{s}\in\tilde{\mathscr{V}}^\mathrm{s}(1/2,1,1)$, and
$V^\mathrm{l}\in\mathscr{V}_D^\mathrm{l}(1/4)$.
Then
\begin{equation}
\int_{-\infty}^\infty
\|\{V^\mathrm{l}(x)-V^\mathrm{l}(pt-e_1t^2/2)\}U_{D,D}^\mathrm{s}(t)\varPhi_v\|\,dt
=O(|v|^{-(4\gamma_D-1)+\epsilon})
\end{equation}
holds as $|v|\to\infty$.
\end{lem}

Here we introduce auxiliary wave operators
$$\Omega_{D,D}^{\mathrm{s},\pm}=\slim_{t\to\pm\infty}
e^{itH}U_{D,D}^\mathrm{s}(t)$$
with $H=H_0+V^\mathrm{vs}+V^\mathrm{s}+V^\mathrm{l}$.
Then we see that
$$\Omega_{D,D}^{\mathrm{s},\pm}=W_D^\pm I_D^{\mathrm{s},\pm},\quad
I_D^{\mathrm{s},\pm}=\slim_{t\to\pm\infty}M_D^\mathrm{s}(t)$$
exist.

\begin{lem}\label{lem4.11}
Let $v$ and $\varPhi_v$ be as in Theorem \ref{thm4.2}, and $\epsilon>0$.
Then
\begin{equation}
\begin{split}
&\sup_{t\in\boldsymbol{R}}\|(e^{-itH}\Omega_{D,D}^{\mathrm{s},-}-U_{D,D}^\mathrm{s}(t))\varPhi_v\|\\
={}&O(|v|^{\max\{-1,-2(\gamma_2-1)+\epsilon,-2\{(\gamma_1-1)+\gamma_D\}+\epsilon,-(4\gamma_D-1)+\epsilon\}})
\end{split}
\end{equation}
holds as $|v|\to\infty$ for $V^\mathrm{vs}\in\mathscr{V}^\mathrm{vs}$,
$V^\mathrm{s}\in\tilde{\mathscr{V}}^\mathrm{s}(1/2,1,1)$,
and $V^\mathrm{l}\in\mathscr{V}_D^\mathrm{l}(1/4)$.
\end{lem}

Based on Lemmas \ref{lem4.8}, \ref{lem4.9}, \ref{lem4.10} and \ref{lem4.11},
Theorem \ref{thm4.2} can be shown by the additional conditions
\begin{equation}
\gamma_2>5/4,\quad
\gamma_1+\gamma_D>5/4,\quad
\gamma_D>3/8
\end{equation}
in the same way as in the proof of Theorem \ref{thm3.1}.
So we omit its proof.
Here we note that $\gamma_1+\gamma_D>5/4$ is satsified
even when $\gamma_1>1$ and $\gamma_D>1/4$.
\bigskip

\leftline{\textbf{Acknowledgement}}
The first author
is partially supported by the Grant-in-Aid for Scientific Research
(C) \#17K05319 from JSPS.

\end{document}